\documentclass[sn-basic]{sn-jnl}

\usepackage[skip=2pt, margin=90pt, subrefformat=parens,labelformat=parens]{subcaption}
\usepackage{multirow}%
\usepackage{arydshln}
\jyear{2021}%

\theoremstyle{thmstyleone}%
\newtheorem{theorem}{Theorem}
\newtheorem{proposition}[theorem]{Proposition}%

\theoremstyle{thmstyletwo}%

\theoremstyle{thmstylethree}%

\begin{document}

\title[ The shape of the relative frailty variance ...]{\centering The shape of the relative frailty variance induced by discrete random effect distributions in univariate and multivariate survival models \\
}


\author*[1]{\fnm{Maximilian} \sur{Bardo}}\email{maximilian.bardo@med.uni-goettingen.de}

\author[1,2]{\fnm{Steffen} \sur{Unkel}}\email{steffen.unkel@med.uni-goettingen.de}


\affil[1]{\orgdiv{Department of Medical Statistics}, \orgname{University Medical Center G\"{o}ttingen}, \orgaddress{
\country{Germany}}} 

\affil[2]{\orgdiv{Faculty V: School of Life Sciences}, \orgname{University of Siegen}, \orgaddress{
\country{Germany}}}



\abstract{In statistical models for the analysis of time-to-event data, individual heterogeneity is usually accounted for by means of one or more random effects, also known as frailties. In the vast majority of the literature, the random effect is assumed to follow a continuous probability distribution.
However, in some areas of application, a discrete frailty distribution may be more appropriate.
We investigate and compare various existing families of discrete univariate and shared frailty models by taking as our focus the variance of the relative frailty distribution in survivors. The relative frailty variance ($\text{RFV}$) among survivors provides a readily interpretable measure of how the heterogeneity of a population, as represented by a frailty model, evolves over time. We explore the shape of the $\text{RFV}$ for the purpose of model selection and review available discrete random effect distributions in this context. 
We find non-monotone trajectories of the $\text{RFV}$ for discrete univariate and shared frailty models, which is a rare property. 
Furthermore, we proof that for discrete time-invariant univariate and shared frailty models with (without) an atom at zero, the limit of the $\text{RFV}$ approaches infinity (zero), {if the support of the discrete distribution can be arranged in ascending order}. 
Through the {one-to-one} relationship of the $\text{RFV}$ with the cross-ratio function in shared frailty models, which we generalize to the higher-variate case, our results also apply to {patterns of association within a cluster}.
Extensions and contrasts to discrete time-varying frailty models and contrasts to correlated discrete frailty models are discussed.}




\maketitle

\section{Introduction}\label{intro}


In the modelling of time-to-event data, unobserved heterogeneity among observational units is commonly accounted for by random effects, 
also known as frailties \citep[see, e.g.,][]{Aalen.2008,Duchateau.2008,Hougaard.2000,Wienke.2010}. 
Frailty models provide a conceptually appealing way of quantifying the association between
clustered event times and of representing heterogeneity across observational units resulting from factors that
may be difficult or impossible to measure.
The vast majority of the literature treats the random effect as being continuous; the gamma, inverse Gaussian, and log-normal distributions are typical choices for the frailty (see, e.g., the books cited above and the references therein). 
However, in certain areas of application, some of which are mentioned further below, discrete frailty distributions may be more appropriate.

Discrete frailty distributions for univariate and multivariate data have been considered by \cite{Pickles.1994}, \cite{Begun.2000}, \cite{Caroni.2010}, Section 3.2 in \cite{Wienke.2010} and the references therein, \cite{Farrington.2012}, \cite{Ata.2013}, \cite{Bijwaard.2014}, \cite{Palloni.2017}, \cite{Souza.2017}, \cite{Cancho.2018,Cancho.2020}, \cite{Gasperoni.2020}, \cite{Mohseni.2020}, \cite{Cancho.2021}, \cite{Molina.2021}, and \cite{Cancho.2020b}. 
Among the utilized distributions in this setting are the geometric distribution \citep{Caroni.2010,Cancho.2021}, 
the negative binomial distribution, the Poisson distribution \citep{Caroni.2010},
the zero-inflated power series distribution which contains the aforementioned distributions as special cases \citep{Cancho.2018,Cancho.2020},
the zero-modified power series distribution which contains the zero-inflated power series distribution as a special case \citep{Molina.2021}, and the hyper-Poisson distribution \citep{Souza.2017,Mohseni.2020}. 
\cite{Ata.2013} introduced a discrete compound Poisson process in which the Poisson distribution determines the number of random variables that are either binomial, negative binomial, or Poisson, and in which these random variables add up to the frailty.
\cite{Farrington.2012} introduced the Addams family of discrete frailty distributions, which includes the scaled negative binomial, binomial, Poisson, and, as a continuous exception, the gamma distribution.
Furthermore, the $k$-point distribution, i.e. a distribution for which the finite support of the frailty is multinomial distributed with $k$ categories, has also been proposed; see, for example, Section 3.2 in \cite{Wienke.2010} and the references therein, \cite{Begun.2000}, and \cite{Pickles.1994}. 
For a discussion on univariate and multivariate survival data, recurrent events, competing risks data as well as shared and correlated frailties in the context of the $k$-point distribution, see \cite{Bijwaard.2014}.  
A non-parametric version of the $k$-point frailty distribution, where $k$ is also subject to estimation, was suggested by \cite{Gasperoni.2020}.
The $k$-point distribution has been suggested for scenarios in which the unobserved heterogeneity is due to genetic factors \citep[see][]{Pickles.1994,Begun.2000,Wienke.2010}.
Further applications are time to first readmission of heart failure patients to hospitals, where unobserved heterogeneity is assumed to be induced by the performance of health care providers \citep{Gasperoni.2020}.
Labor market applications are suggested by \cite{Bijwaard.2014}.
\cite{Palloni.2017} suggested a delayed binary frailty model, i.e. an excess hazard model where the frailty-weighted excess hazard steps in from some point in time after the start of follow-up.
The delayed frailty effect, which is governed by a continuous random variable, is supposed to represent early-life conditions that affect adult health and mortality.
A discrete frailty model with a correlated instead of a shared frailty has also been suggested by \cite{Cancho.2020b}, utilizing the Poisson and the gamma distributions as components of the latent structure. 

Apart from the $k$-point distribution, the aforementioned discrete distributions have support on the natural numbers including zero. 
If the frailty acts multiplicatively on the baseline hazard rate, the positive probability at zero leads to
observational units not being at risk of experiencing the event(s) of interest, i.e. they are cured \citep[p. 244]{Aalen.2008}.
Hence, a proposed area of application for discrete frailty models are data for which a cure rate is appropriate.
For applications of discrete cure rate models to time-to-death or time-to-recurrence in melanoma data, see \cite{Souza.2017},\cite{Cancho.2018,Cancho.2020}, and \cite{Molina.2021}.
For an application of those models to time-to-death from Hodgkin lymphoma, see \cite{Mohseni.2020}.

\cite{Caroni.2010} suggest {univariate} discrete frailty models without a cure rate by adding
an additional hazard rate on the frailty-weighted baseline hazard.
An application in reliability analysis can be found in \cite{Caroni.2010}. 
\cite{Caroni.2010} and \cite{Molina.2021} also suggested non-cure-rate models by truncating zero from the support of the discrete frailty distribution.

The shapes of association for clustered event times and patterns of heterogeneity across observational units in discrete frailty models have not been investigated thoroughly yet.
Only \cite{Farrington.2012} studied the heterogeneity induced by the Addams family of discrete frailty distributions, and \cite{Cancho.2020b} examined the association within a cluster induced by their proposed correlated frailty model.
Therefore, the present paper investigates the shapes of association and patterns of heterogeneity, respectively, of discrete time-invariant univariate and shared frailty models.
We examine some of the aforementioned models and then proceed to a general investigation of discrete time-invariant univariate and shared frailty models in this respect.
This will be done by studying the variance of the relative frailty distribution in survivors \citep{Hougaard.1984}, i.e. the relative frailty variance ($\text{RFV}$), which is a local measure of heterogeneity across observational units over time.
For bivariate shared frailty models, the $\text{RFV}$ has a one-to-one relationship with the cross-ratio function ($\text{CRF}$) \citep{Anderson.1992}, which is a local measure of association within a cluster over time.
We slightly adapt the definition of the $\text{CRF}$ to retain the relationship with the $\text{RFV}$ for cases with more than two time-to-event variables per observational unit.
Thus, in a shared frailty model, our results concerning the $\text{RFV}$ also apply to the (adapted) $\text{CRF}$.
The advantage of using the $\text{RFV}$ compared to the $\text{CRF}$ is that it can also be utilized in the univariate case, i.e. for scenarios with a single time-to-event variable per observational unit, as a measure of heterogeneity.
We show that for discrete time-invariant univariate and shared frailty models with (without) an atom at zero, the limit of the $\text{RFV}$ is infinity (zero).
We also extend this finding to piecewise constant discrete time-varying univariate and shared frailty models.
{As a contrast, we show that the models introduced by \cite{Caroni.2010}, which do not follow our definition of time-invariant univariate frailty models, do not share this property.}
In the multivariate context, we contrast our findings on discrete shared frailty models with the class of discrete correlated frailty models introduced by \cite{Cancho.2020b}.

The remainder of the present paper is organized as follows. Section \ref{Frailty} defines the time-invariant univariate and shared frailty model that we consider throughout this paper. 
Moreover, the concept of the $\text{RFV}$ and its relation to the $\text{CRF}$ is revisited.
Section \ref{Shapes} investigates the $\text{RFV}$ of selected discrete univariate and shared frailty models. 
A general investigation of the $\text{RFV}$ for discrete univariate and shared frailty models follows.
The section concludes with extensions and contrasts of the former results to time-varying univariate and shared, as well as correlated frailty models.
Concluding remarks are given in Section \ref{conc}.
Further technical details and proofs are given in the Appendices.


\section{Univariate and shared frailty models for the hazard rates} \label{Frailty}



Consider the time-to-event random vector (RVe) $\boldsymbol{T} =[T^{(1)},\dots,$ $ T^{(J)}]^T$, with $J \geq1 $, and $T^{(j)} \in \mathbb{R}_{>0} \ \forall j$. 
We will refer to a specific $T^{(j)}$ from $\boldsymbol{T}$ as a target variable.
Let $\boldsymbol{x}^{(j)}(t)$, with $t \in \mathbb{R}_{\geq 0}$, and $j=1,\dots, J$, denote the realizations of observable, possibly time-dependent covariate vectors which are allowed to vary across $j$, $j=1,\dots,J$. 
Given covariates, the RVe $\boldsymbol{T}$ follows a (multivariate) population time-to-event distribution $\pi_{\boldsymbol{T}}$ with distribution parameters $\boldsymbol{\theta}_{\boldsymbol{T}}$. 
Covariates might affect the time-to-event distribution $\pi_{\boldsymbol{T}}$ itself and/or the distribution parameters $\boldsymbol{\theta}_{\boldsymbol{T}}$.
For different observational units $i$ and $i'$ and conditional on covariates, $\boldsymbol{T}_i$ is assumed to be independent of $\boldsymbol{T}_{i'}$. 
In the multivariate context, i.e. for $J \geq 2$, $T^{(j)}$ and $T^{(j')}$, with $j \neq j'$, are assumed to be independent within an observational unit (cluster) given unobservable cluster-specific characteristic $z$ and covariates.
The unit-specific latent characteristic $z$ is assumed to be a realization of the random variable (RVa) $Z\geq 0$, which follows a distribution $\pi_Z$ with parameter vector $\boldsymbol{\theta}_Z$.
The distribution $\pi_Z$ and/or the distribution parameters $\boldsymbol{\theta}_Z$ might depend on (a subset of) unit-specific covariates at the start of follow-up,
i.e. at time $t=0$, which are gathered in the vector $\tilde{\boldsymbol{x}}$.
The draws from $Z$ are assumed to be independent across observational units given $\tilde{\boldsymbol{x}}$.
Note that this, for example, encompasses random slopes for time- and target-invariant covariates, i.e. $z =z_1\exp\{\ln\{\boldsymbol{z}^T_2\}\tilde{\boldsymbol{x}}\}$, where $[z_1, \boldsymbol{z}^T_2]$ is an unobservable, unit-specific realization of a non-negative RVe, which is independent and identically distributed across observational units.

The conditional target-specific hazard rate, which determines the conditional time-to-event distribution, is assumed to be of the form
\begin{align}
\lambda^{(j)}(t\mid Z=z) &= \lim_{\Delta \to 0}\frac{P(t \leq T^{(j)} < t + \Delta \mid  T^{(j)} \geq t, Z = z,\boldsymbol{x}^{(j)}(t))}{\Delta} \nonumber \\ 
&= z \lambda_x^{(j)}(\boldsymbol{x}^{(j)}(t),t), \label{model}
\end{align}
with covariate and target-specific hazard rate $\lambda_x^{(j)}(\boldsymbol{x}^{(j)}(t),t)> 0$, and $t \in\mathbb{R}_{> 0}$.
We call \eqref{model} a time-invariant frailty model, as $z$ does not vary with $t$.
If $J>1$, \eqref{model} is a time-invariant shared frailty model, i.e. $z$ does not differ across $j=1,\dots,J$.

The unit-specific latent characteristic $z$ in equation \eqref{model} is also called a random effect (RE) or frailty hereafter.
The frailty induces association in the time-to-event targets within a cluster and heterogeneity across clusters or observational units, respectively.
We will use the terms (shared) frailty model, time-invariant (shared) frailty model, or RE model synonymously throughout this paper and refer to equation \eqref{model} in this case. 
We will restrict attention to time-invariant (shared) frailty models unless otherwise stated.

In the following, we assume that all quantities are conditioned on potential covariates without explicitly mentioning that in the notation.
Hence, we will drop covariates from the notation of the hazard rates and set $\lambda_x^{(j)}(\boldsymbol{x}^{(j)}(t),t)$ equal to the baseline hazard rate $ \lambda_0^{(j)}(t)$.

In the case of multivariate time-to-event data, i.e. for $J\geq2$, the $\text{CRF}$  \citep{Clayton.1978} is a suitable measure of association within a cluster.
The $\text{CRF}$ was originally introduced as a measure of bivariate association, and can be expressed as follows:
\begin{align}
\text{CRF}^{(j,j')}(\boldsymbol{t}) = \frac{h^{(j)}(t_{j}\mid T^{(j')}=t_{j'})}{h^{(j)}(t_j\mid  T^{(j')}>t_{j'})},\label{biCRF}
\end{align}
where $\boldsymbol{t} = [t_{j},t_{j'}]^T$, and $h^{(j)}(t\mid A)$ denotes the population hazard rate with observable condition $A$, 
i.e. $ \lim_{\Delta \to 0}\frac{P(t\leq T^{(j)} < t+\Delta\mid T^{(j)}\geq t,A)}{\Delta}$.
In the case of a time-invariant shared frailty model, the $\text{CRF}$ is symmetric in its superscript, i.e. $\text{CRF}^{(j,j')}(\boldsymbol{t}) = \text{CRF}^{(j',j)}(\boldsymbol{t})$,
as the $\text{CRF}$ can be re-expressed as a function of $\Lambda_0^{(j,j')}(\boldsymbol{t}) = \Lambda_0^{(j)}(t_j) + \Lambda_0^{(j')}(t_{j'}) = \Lambda_0^{(j',j)}(\boldsymbol{t})$, 
with $\Lambda_0^{(j)}(t)  = \int_0^t \lambda^{(j)}_0(u)du$ \citep{Oakes.1989}.
However, $\text{CRF}^{(j,j')}(\boldsymbol{t}) \neq \text{CRF}^{(j,j'')}(\boldsymbol{t})$ for $\boldsymbol{t}=[t_j,t]^T$, if $j' \neq j''$ and $\Lambda_0^{(j')}(t) \neq \Lambda_0^{(j'')}(t)$ as the implicit argument of the $\text{CRF}$ changes from $\Lambda_0^{(j,j')}(\boldsymbol{t})$ to $\Lambda_0^{(j,j'')}(\boldsymbol{t})$.
The reason for the inequation of the former two $\text{CRF}$s is not that the pattern of association differs across the two pairs $(j,j')$ and $(j,j'')$, which is determined by the cluster-specific frailty $z$ for both tuples, but that the generic time as defined by the cumulative baseline hazards that determine the selection process of the frailty \citep{Farrington.2012}, proceeds differently with $t$.
Hence, we redefine the $\text{CRF}$ to avoid the former inequation of the $\text{CRF}$s in order to obtain a measure of association that is also suitable for the higher variate case.
Let
\begin{align*}
\text{CRF}^{(j,j')}(\boldsymbol{t}) = \frac{h^{(j)}(t_{j}\mid T^{(j')}=t_{j'}, T^{(l)} >t_l\ \forall\ {l \neq j,j'})}{h^{(j)}(t_{j}\mid  T^{(l)} > t_l\ \forall\ {l\neq j})},
\end{align*}
where the population hazard rate is the same as equation \eqref{biCRF}, except that the observable condition $A$ in ${h^{(j)}({t \mid A})}$ is expanded to an arbitrary number of time-to-event RVas and $\boldsymbol{t}=[t_{1},\dots,t_{J}]^T$ accordingly. 
If the $\text{CRF}$ is considered as a function of the sum of the cumulative baseline hazards, i.e. $\text{CRF}^{(j,j')}(\boldsymbol{t})=\text{CRF}^{(j,j')}(\Lambda_0(\boldsymbol{t}))$, where $\Lambda_0(\boldsymbol{t}) = \sum_{j=1}^J \Lambda_0^{(j)}(t_j)$,
then $\text{CRF}^{(j,j')}(\boldsymbol{t}) = \text{CRF}^{(j,j'')}(\boldsymbol{t})$. 
This holds even if $j' \neq j''$ and $\Lambda^{(j')}(t) \neq \Lambda^{(j'')}(t)$ for some or all $t$, as the argument $\Lambda_0(\boldsymbol{t})$ is unaffected by the choice of the tuples $(j,j')$ and $(j,j'')$.
As the $\text{CRF}$ does not depend on the choice of $j$ and $j'$, we drop the superscripts in the following.
This adaption makes the $\text{CRF}$ a multivariate measure of association within a cluster on the generic time scale $\Lambda \in \{\Lambda_0(\boldsymbol{t}): \boldsymbol{t} \in \mathbb{R}_{\geq 0} \}$.
Note that $\Lambda_0^{(j)}(0)$ is defined to be zero for all $j$.

This definition of the $\text{CRF}$ retains $\text{CRF}(\boldsymbol{t})=1+\text{RFV}(\boldsymbol{t})$, with relative frailty variance $\text{RFV}(\boldsymbol{t}) = \frac{\operatorname{Var}(Z\mid \boldsymbol{T}>\boldsymbol{t})}{\operatorname{E}(Z\mid \boldsymbol{T}>\boldsymbol{t})^2}$, known from the bivariate case \citep{Anderson.1992} for situations 
with more than two time-to-event RVas per cluster.
The derivation is basically identical to the bivariate case and can be found in Appendix \ref{CRFtoRFV}. 
The advantage of the $\text{RFV}$ compared to the $\text{CRF}$ is that it also exists for univariate time-to-event data ($J=1$).
In the univariate as well as the multivariate case, the $\text{RFV}$ might be interpreted as the dispersion of observational unit-specific, conditional hazard rates relative to the population hazard, given an event-free cluster up to $\boldsymbol{t}$ in the multivariate case.
More precisely, 
\begin{align*}
\text{RFV}(\boldsymbol{t}) &= \frac{\lambda_0^{(j)}(t_j)^2 \operatorname{Var}(Z\mid \boldsymbol{T} > \boldsymbol{t})}{\lambda_0^{(j)}(t_j)^2 \operatorname{E}(Z\mid \boldsymbol{T} > \boldsymbol{t})^2} \\
&= \frac{\operatorname{Var}(\lambda^{(j)}(t_j\mid Z)\mid T^{(l)}> t_l\ \forall l \neq j)}{h^{(j)}(t_j\mid T^{(l)}> t_l\ \forall l \neq j)^2},
\end{align*}
which does not depend on the choice of $j$ in a time-invariant shared frailty model. 
The one-to-one relationship between the $\text{CRF}$ and the $\text{RFV}$ shows, firstly, that the $\text{RFV}$ can also be expressed as a function of the generic time scale $\Lambda$.
Secondly, statements about the heterogeneity as measured by the $\text{RFV}$ lead to equivalent statements about association as measured by the $\text{CRF}$ for shared frailty models.


The $\text{RFV}$ uniquely determines the Laplace transform of $Z$, $\mathcal{L} (s) = \int_0^\infty  \exp\{ -z s\} g(z)dz$, and, therefore, the density $g(z)$ of $Z$ as well as the (multivariate) population time-to-event survival function $S(\boldsymbol{t}) = P(\boldsymbol{T}>\boldsymbol{t})= \mathcal{L}(\Lambda(\boldsymbol{t}))$, up to the mean $\operatorname{E}(Z)$ \citep{Farrington.2012}.
This highlights that the $\text{RFV}$ is an important quantity that should be considered in the model-building process, as a misspecified $\text{RFV}$ leads to a misspecified model and possibly to invalid inference.
To the best of our knowledge, the $\text{RFV}$ has not been investigated for discrete frailty models thoroughly yet.
Hence, in the next section, we investigate the trajectories of the $\text{RFV}$ of numerous discrete frailty models as well as the {long-term} shapes of the $\text{RFV}$ that discrete frailty models are able to induce in general.

\section{Shapes of the RFV} \label{Shapes}



\subsection{Families of discrete frailty distributions and their RFV}
Among the discrete frailty distributions that we consider in this section are the negative binomial ($\mathcal{NB}$), binomial ($\mathcal{B}$), Poisson ($\mathcal{P}$), { as well as $k$-point} distributions.
We also consider the zero-modified Poisson ($\mathcal{ZMP}$) distribution in which the probability at $Z=0$ is modelled by an additional deflation/inflation parameter.
We further distinguish between the $\mathcal{NB}$ for which zero is an element of the support, i.e. with underlying RVa being the ``number of failures until $\nu$ successes are observed",
and the $\mathcal{NB}_{>0}$ for which zero is not an element of the support, i.e. with underlying RVa being the ``number of trials until $\nu$ successes are observed". 

Discrete frailty models in the literature are often used as cure rate models, i.e. $P(Z=0)>0$ \citep[see][]{Caroni.2010,Ata.2013,Souza.2017,Cancho.2018,Cancho.2020,Mohseni.2020,Cancho.2021,Molina.2021}. 
Discrete frailty models without cure rate were introduced by \cite{Caroni.2010} in a univariate setting, which can be adapted to the multivariate case as follows: $\lambda^{(j)}(t \mid Z=z) = z \lambda^{(j)}_0(t)+\lambda^{(j)}_+(t)$, with $\lambda^{(j)}_+(t) \in \mathbb{R}_{\geq 0}$ for all $t$. 
Expressed with a multiplicative hazard, the model might be depicted as $\lambda^{(j)}(t \mid Z=z) = (z +  p^{(j)}(t))\lambda^{(j)}_0(t)$, with $p^{(j)}(t)=\frac{\lambda^{(j)}_+(t)}{\lambda_0^{(j)}(t)}$.
This model does not correspond to equation \eqref{model} of a time-invariant (shared) frailty model in general, as  the frailty depends on time (and $j$).
With the proportionality restriction $p=\frac{\lambda^{(j)}_+(t)}{\lambda_0^{(j)}(t)}$ for $j=1,\dots,J$ and all $t$, the model is a (shared) frailty model according to equation \eqref{model}, i.e. the frailty neither depends on time $t$ nor on $j$. 
The frailty distribution is then some (discrete) distribution with support shifted by $p \in \mathbb{R}_{\geq 0}$.
The discrete (shared) frailty models of this kind that we investigate are the shifted negative binomial ($\mathcal{NB}+p$), shifted binomial ($\mathcal{B}+p$) and the shifted Poisson ($\mathcal{P}+p$).
We call frailty distributions, where the support is shifted by the additional parameter $p$, shifted models.

\cite{Caroni.2010} also suggested to truncate zero from the support of the discrete frailty distribution. 
This idea was taken up by \cite{Molina.2021}, where the truncation is implicitly considered by a special case of the zero-modified frailty distribution
and hence, we consider such a model through the $\mathcal{ZMP}$ distribution.
We call models, where zero is truncated from the support of the discrete frailty distribution, truncated models.
We call models, where the probability of $P(Z=0)$ is either deflated or inflated by means of an additional parameter, zero-modified models.

The use of the shifted and the zero-modified model have a drawback.
On the one hand, if zero is accommodated by the shifted model, the parameter $p$ has to be on the edge of the parameter space to allow for $P(Z=0)>0$.
On the other hand, in the zero-modified model, the deflation/inflation parameter has to be on the edge of the parameter space to correspond to the truncated model, i.e. for $P(Z=0)=0$.
This might lead to difficulties in the estimation process if one of those cases is the ``true" data-generating process or fits best to the data at hand.
We will show that even very small values for $p$ in the shifted model or $P(Z=0)$ in the zero-modified model, respectively, have an enormous influence on the $\text{RFV}$.

\cite{Farrington.2012} introduced the Addams family ($\mathcal{AF}$) of discrete frailty distributions. 
The $\mathcal{AF}$ circumvents the above-mentioned problem that either $P(Z=0)>0$ or $P(Z=0)=0$ is on the edge of the parameter space.
For this family it holds that $\text{RFV}(\boldsymbol{t}) = \gamma \exp \{\alpha \Lambda_0(\boldsymbol{t})\}$, with $\gamma>0$, and $\alpha \in \mathbb{R}$. 
If $\alpha < 0, Z$ is distributed as a scaled  $\mathcal{NB}_{>0}$. 
If $\alpha>0, Z$ is either a scaled $\mathcal{NB}$, $\mathcal{B}$ or $\mathcal{P}$, depending on $\mathcal{\gamma}$.
The case $\alpha = 0$ is a continuous exception, $Z$ is then gamma ($\mathcal{G}$) distributed.
Hence, the $\mathcal{AF}$ might be utilized as an exploratory tool, whether zero should be an element of the support of the discrete frailty distribution or not.
However, \cite{Farrington.2012} only introduced the $\mathcal{AF}$ of distributions; there are no estimation procedures available yet.
We recapitulate the shapes of the $\text{RFV}$ of the $\mathcal{AF}$ in this paper.

Before we start investigating the trajectory of the $\text{RFV}$, we will discuss the Laplace transform and $\text{RFV}$ for shifted and zero-modified models.
We start with the shifted model and denote $Z=Z_* + p$, where $Z_*$ induces the randomness in $Z$ and $p$ is the support shifting parameter.
Let $\mathcal{L}_*(s)$ be the Laplace transform of $Z_*$. 
Then, the Laplace transform of $Z$ is $\mathcal{L}(s)=\mathcal{L}_*(s) \exp\{-ps\}$.
The $\text{RFV}$ of the shifted model is $\text{RFV}(\boldsymbol{t}) = \text{RFV}_*(\boldsymbol{t}) \big[ \frac{{\mathcal{L}_*}' (\Lambda_0(\boldsymbol{t}))}{{\mathcal{L}_*}'(\Lambda_0(\boldsymbol{t})) - p\mathcal{L}_*(\Lambda_0(\boldsymbol{t}))} \big]^2$, where $\text{RFV}_*(\boldsymbol{t}) = \frac{\operatorname{Var}(Z_* \mid \boldsymbol{T}>\boldsymbol{t})}{\operatorname{E}(Z_* \mid \boldsymbol{T}>\boldsymbol{t})^2} = \frac{{\mathcal{L}_*}''(\Lambda_0(\boldsymbol{t})) {\mathcal{L}_*}(\Lambda_0(\boldsymbol{t}))}{{\mathcal{L}_*}'(\Lambda_0(\boldsymbol{t}))^2}-1$,
and $\mathcal{L}'(s),\mathcal{L}''(s)$ denote the first and second derivative of the Laplace transform, respectively.
Note that using $Z$ instead of $Z_*$ as RVa retains the relationship $\text{RFV}(\boldsymbol{t}) = \text{CRF}(\boldsymbol{t})-1$, which would not be the case for $\text{RFV}_*(\boldsymbol{t})$ as $\mathcal{L}_*(s)$ looses its connection to the population survival function $S(\boldsymbol{t}) = P(\boldsymbol{T}>\boldsymbol{t}) =\mathcal{L}(\Lambda_0(\boldsymbol{t}))$, and hence to the population hazard rate.
Thus, $\text{RFV}(\boldsymbol{t})$ is a suitable measure of association and heterogeneity, whereas $\text{RFV}_*(\boldsymbol{t})$ is not. 
Therefore, we argue that it makes more sense to interpret $p$ as a support-shifting parameter of the frailty distribution instead of being some additive constant in the model.

The frailty probability mass function (pmf) in the zero-modified model is defined as 
\begin{align*}
g(z) =  \begin{cases}
      \phi g_*(0) & \text{for $z = 0$}\\
       \frac{g_*(z)(1-\phi g_*(0))}{1-g_*(0)}& \text{for $z=1,2,\dots$}\\
    \end{cases},       
\end{align*}
where $g_*(z)$ is some discrete reference distribution, for example $\mathcal{P}$, and $\phi \in[0,\frac{1}{g_*(0)})$.
The parameter $\phi$ is the deflation/inflation parameter, which manipulates the probability at $Z=0$; $\phi = 0$ corresponds to the truncated model.
The Laplace transform of the zero-modified model is $\mathcal{L}(s) = \frac{1-\phi g_*(0)}{1-g_*(0)} \mathcal{L}_*(s) + g_*(0)\frac{\phi - 1}{1-g_*(0)}$, 
where $\mathcal{L}_*$ denotes the Laplace transform of the reference distribution.
The $\text{RFV}$ in the zero-modified model is $\text{RFV}(\boldsymbol{t}) = \text{RFV}_*(\boldsymbol{t}) + \frac{\phi - 1}{1 - \phi g_*(0)}g_*(0) \frac{\text{RFV}_*(\Lambda_0(\boldsymbol{t}))+1}{{\mathcal{L}_*}(\Lambda_0(\boldsymbol{t}))}$.
For shifted and zero-modified models, we will replace $\mathcal{L}_*(s)$ and $\text{RFV}_*(\Lambda_0(\boldsymbol{t}))$ by $\mathcal{L}_{\pi}(s)$ and $\text{RFV}_{\pi}(\Lambda_0(\boldsymbol{t}))$, respectively, if we refer to a specific reference distribution $\pi$.

Table \ref{table:RFVs} in Appendix \ref{ShapesApp} provides an overview of the Laplace transform, $\mathcal{L}(s)$, $\text{RFV}(\Lambda)$, ({detailed} information on) its derivative $\text{RFV}'(\Lambda)$, 
and support of $Z$ for $Z \sim \mathcal{NB}$, 
$\mathcal{NB}_{>0}$, 
$\mathcal{B}$, 
$\mathcal{P}$, 
$\mathcal{NB}+p$, 
$\mathcal{B}+p$, 
$\mathcal{P}+p$, 
the $\mathcal{AF}$, 
and the $\mathcal{ZMP}$ distribution.

Graphical illustrations of the resulting $\text{RFV}$ plotted against the generic time scale $\Lambda \in \{\Lambda_0(\boldsymbol{t}): \boldsymbol{t} \in \mathbb{R}^J_{\geq 0}\}$ for the aforementioned (shared) frailty models are given {in the subfigures \subref{fig:RFVploti} to \subref{fig:RFVplotvi} of} Fig. \ref{fig:RFVplot}.
The non-shifted models, $\mathcal{B}, \mathcal{P}$  are indicated by $p=0$ in the legend of Fig. \ref{fig:RFVplot}\subref{fig:RFVplotiii} and \ref{fig:RFVplot}\subref{fig:RFVplotiv}. 
Note that the intercept, steepness, and location of minima and maxima - if existent - of the $\text{RFV}$ depend on the choice of parameter values, but the general trajectories of the $\text{RFV}$ are sufficiently transparent in Fig. \ref{fig:RFVplot}.
{In the following, we will use the term ``tail" to denote the trajectory of either the $\text{RFV}$ from that point on the time axis $\Lambda$ onwards, where the last change in sign of the slope of the $\text{RFV}$ occurs or the entire trajectory if no change in sign of the slope occurs.

The trajectory of the $\text{RFV}$ of the $\mathcal{AF}$ model, which has already been discussed in \cite{Farrington.2012}, can be seen in Fig. \ref{fig:RFVplot}\subref{fig:RFVploti}.
The $\text{RFV}$ is monotone in any case, but can be strictly increasing, strictly decreasing or constant.
The unscaled members of the $\mathcal{AF}$ lead to monotone trajectories of the $\text{RFV}$ 
(see Fig. \ref{fig:RFVplot}\subref{fig:RFVplotii} for the $\mathcal{NB}$ and $\mathcal{NB}_{>0}$, Fig. \ref{fig:RFVplot}\subref{fig:RFVplotiii} for the $\mathcal{B}$, and Fig. \ref{fig:RFVplot}\subref{fig:RFVplotiv} for the $\mathcal{P}$).
As the $\mathcal{AF}$ contains all of those distributions, the trend of the $\text{RFV}$ is not set beforehand but is governed through $\alpha$. 
The $\text{RFV}$  is strictly increasing for $\alpha>0$, i.e. if the frailty distribution is either a scaled $\mathcal{NB},\mathcal{P}$, or $\mathcal{B}$.
In this case, $Z=0$ is an element of the support of the discrete frailty distribution.
If $\alpha <0$, the $\text{RFV}$ is strictly decreasing and the $\text{RFV}$ is a scaled $\mathcal{NB}_{>0}$.
In this case, $Z=0$  is not an element of the support of the frailty distribution.
If the scaling parameter $\alpha$ approaches zero, the frailty distribution approaches the $\mathcal{AF}$'s continuous exception, the $\mathcal{G}$ distribution. 
Thereby, a constant $\text{RFV}$ is a possible shape resulting from a member of the $\mathcal{AF}$, a shape that is impossible for a discrete shared frailty model to generate,
as the $\mathcal{G}$ is the only distribution with constant $\text{RFV}$.}
\begin{figure}
\centering
\begin{subfigure}{0.49\textwidth}
    \includegraphics[scale=.225]{{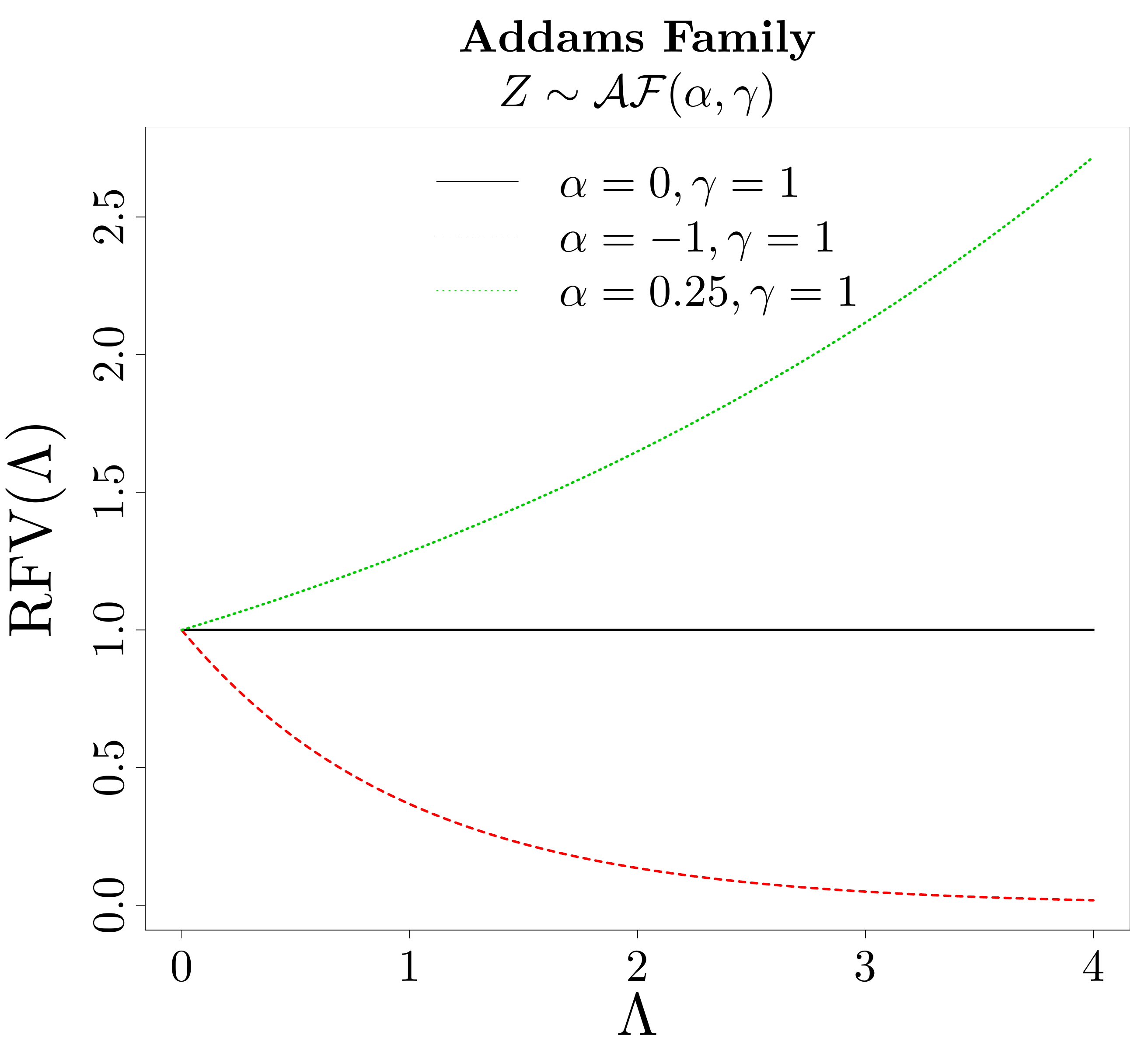}} 
        \caption{}    \label{fig:RFVploti}
\end{subfigure}
\begin{subfigure}{0.49\textwidth}
    \includegraphics[scale=.225]{{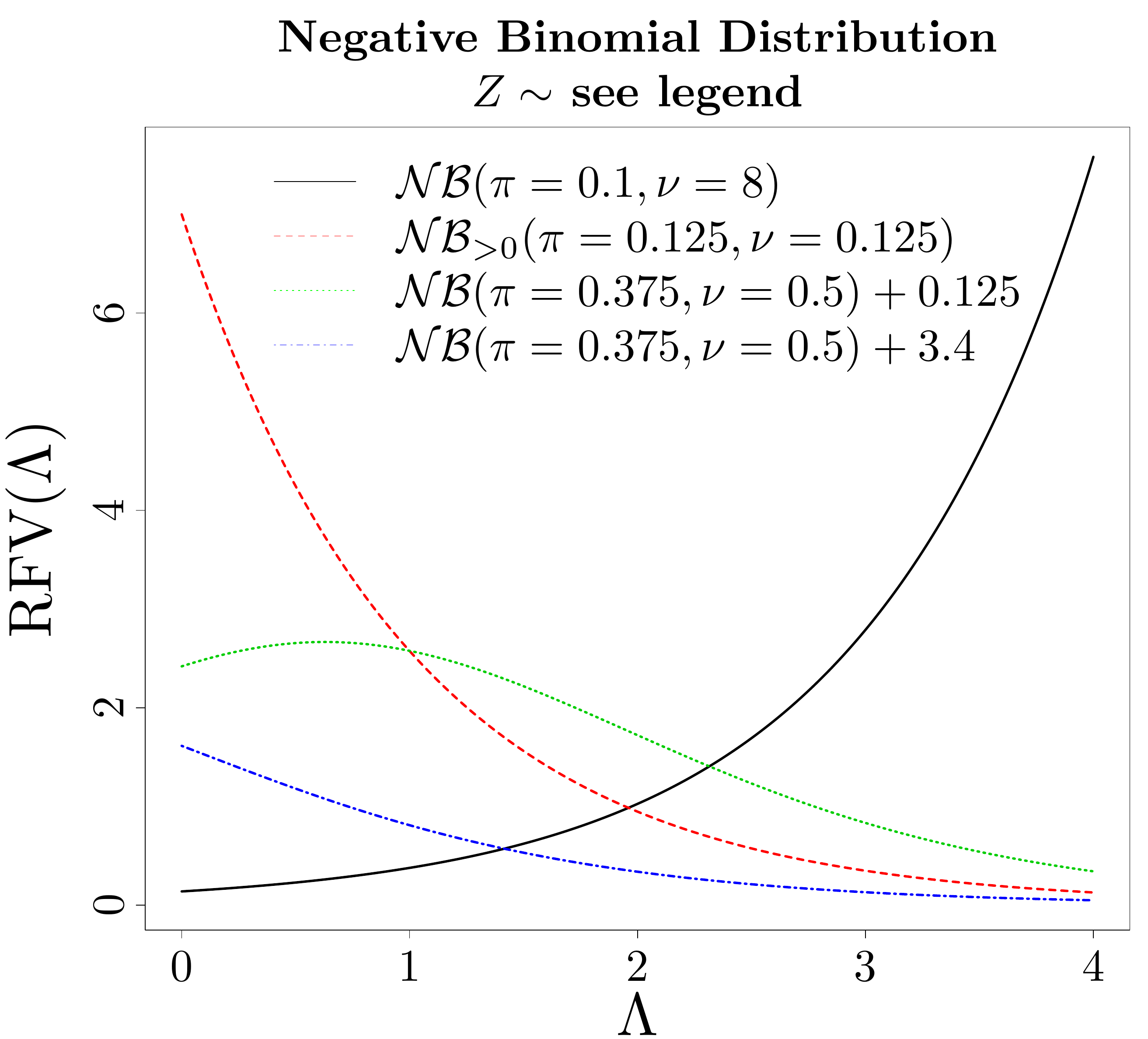}}
        \caption{}    \label{fig:RFVplotii}
\end{subfigure}
\begin{subfigure}{0.49\textwidth}
	\includegraphics[scale=.225]{{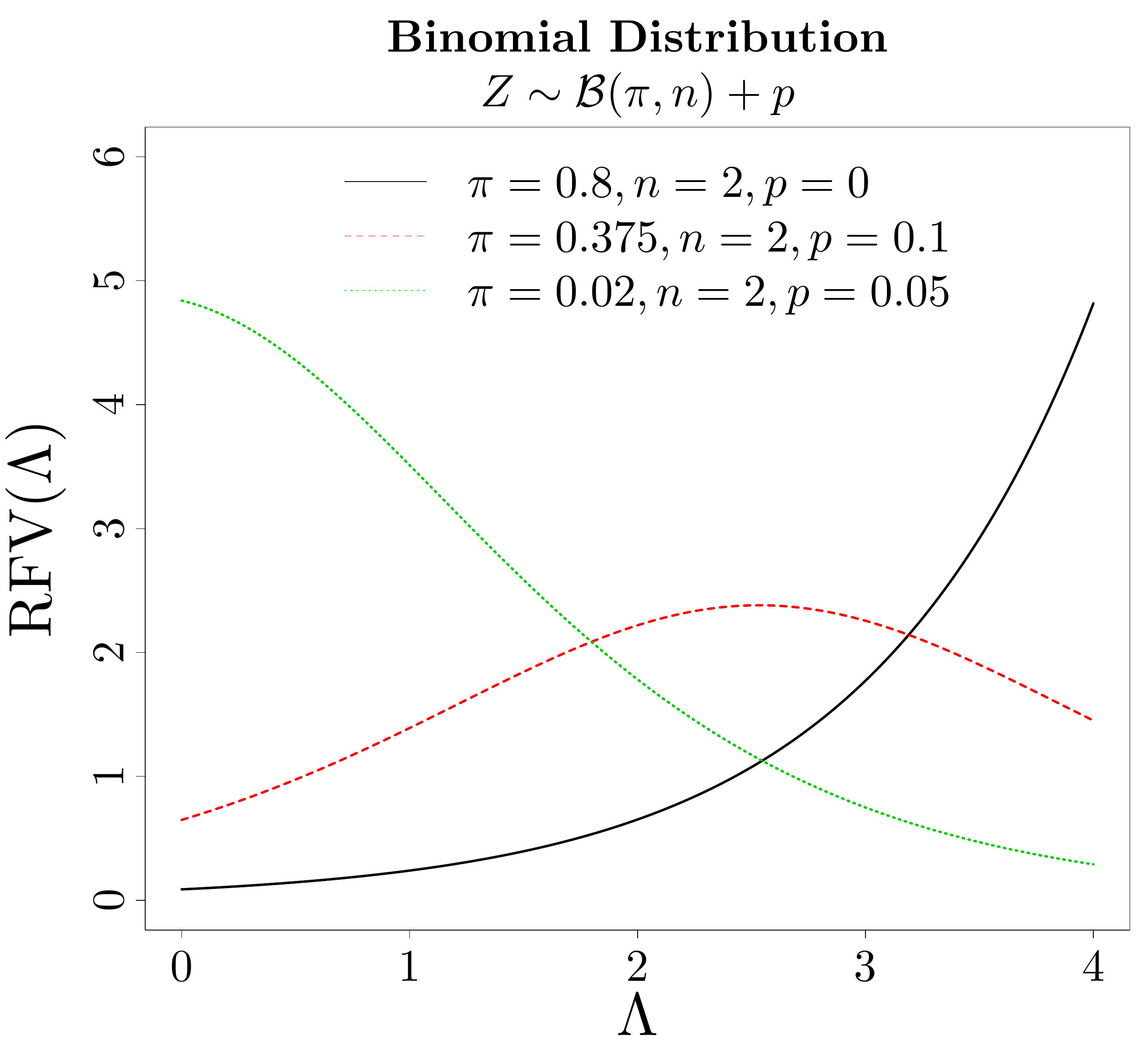}}
         \caption{}     \label{fig:RFVplotiii}
\end{subfigure}
\begin{subfigure}{0.49\textwidth}
    \includegraphics[scale=.225]{{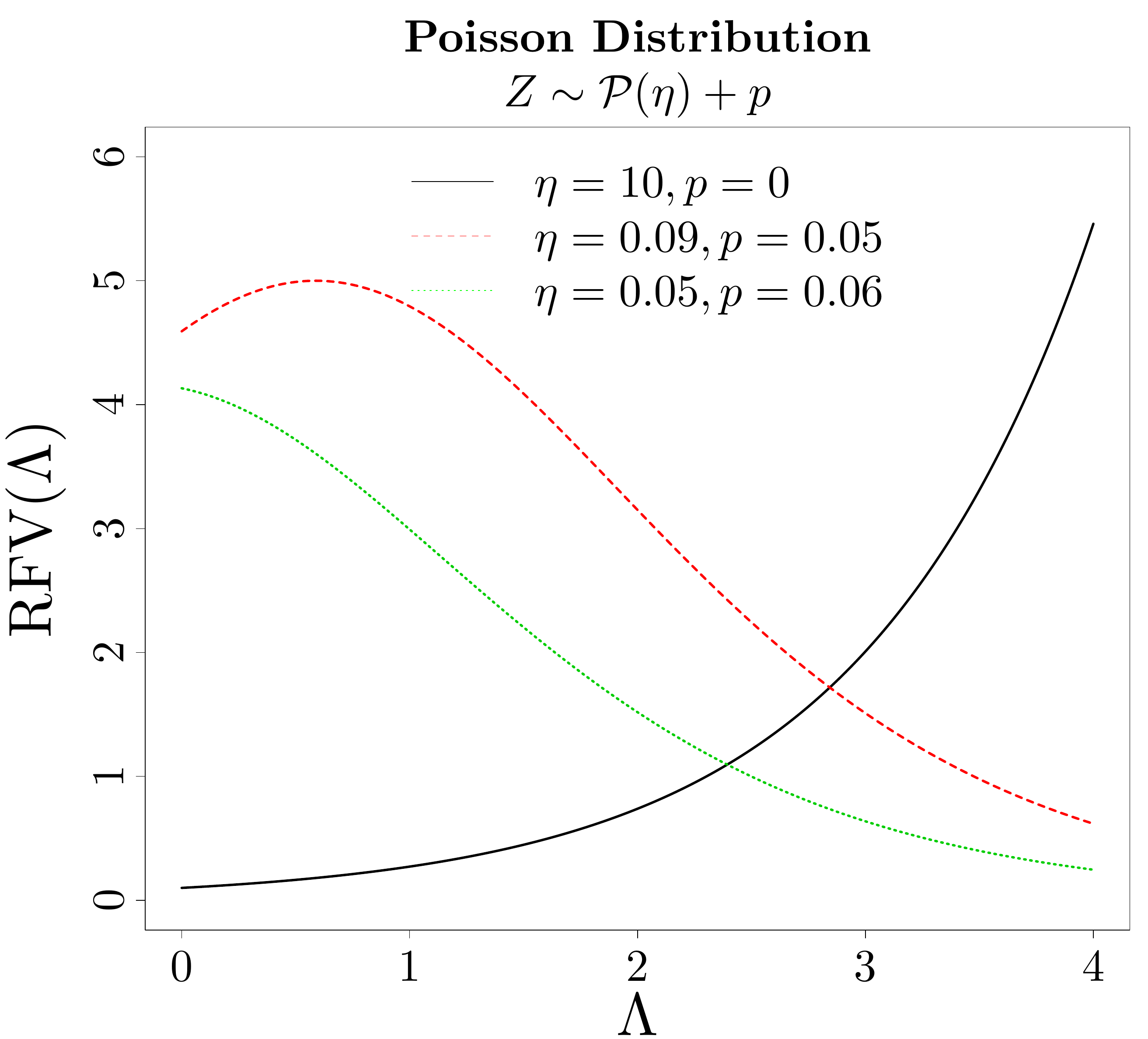}}
    \caption{}    \label{fig:RFVplotiv}
\end{subfigure}
\begin{subfigure}{0.49\textwidth}
    \includegraphics[scale=.225]{{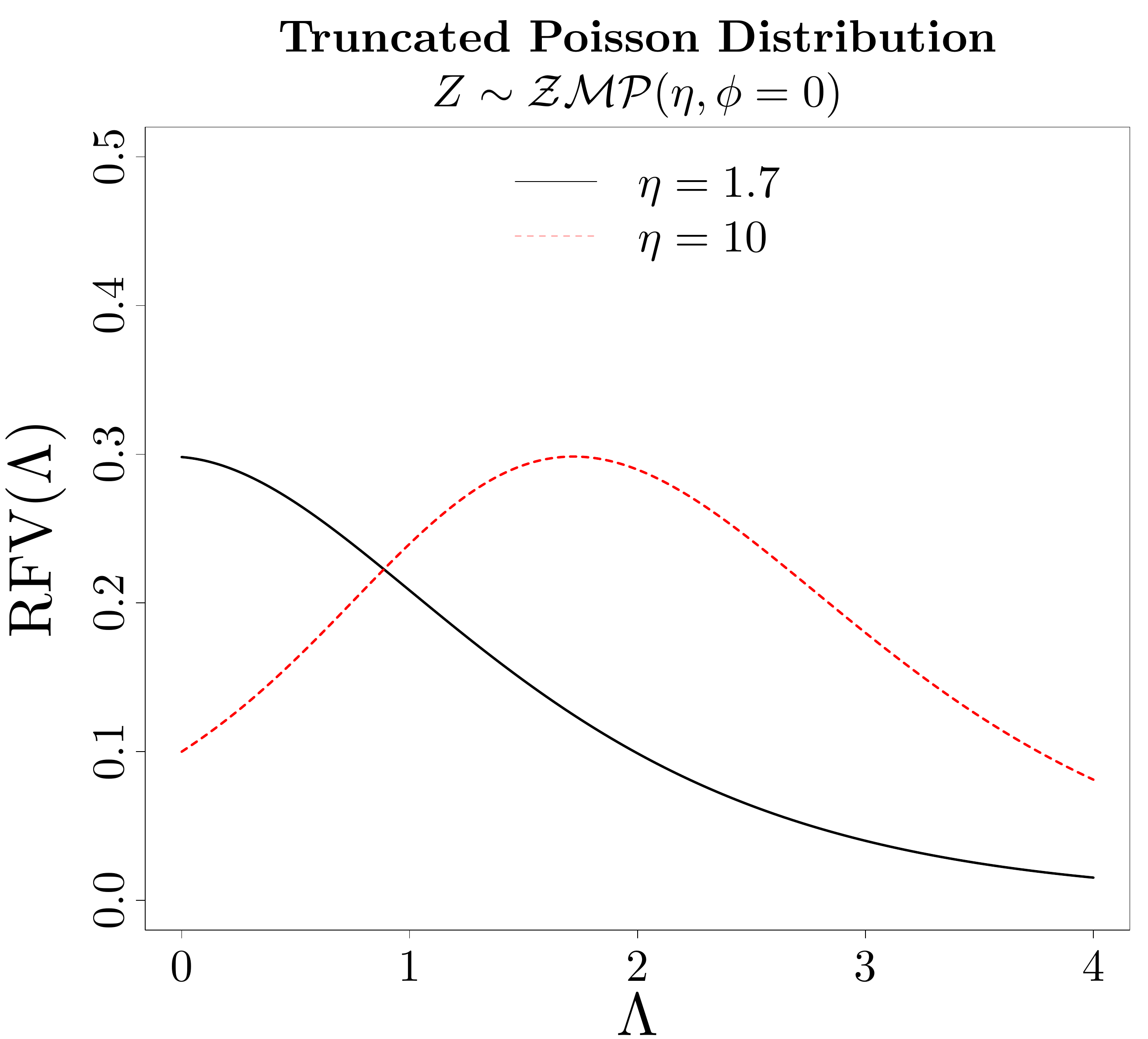}}
       \caption{}   \label{fig:RFVplotv}
\end{subfigure}
\begin{subfigure}{0.49\textwidth}
    \includegraphics[scale=.225]{{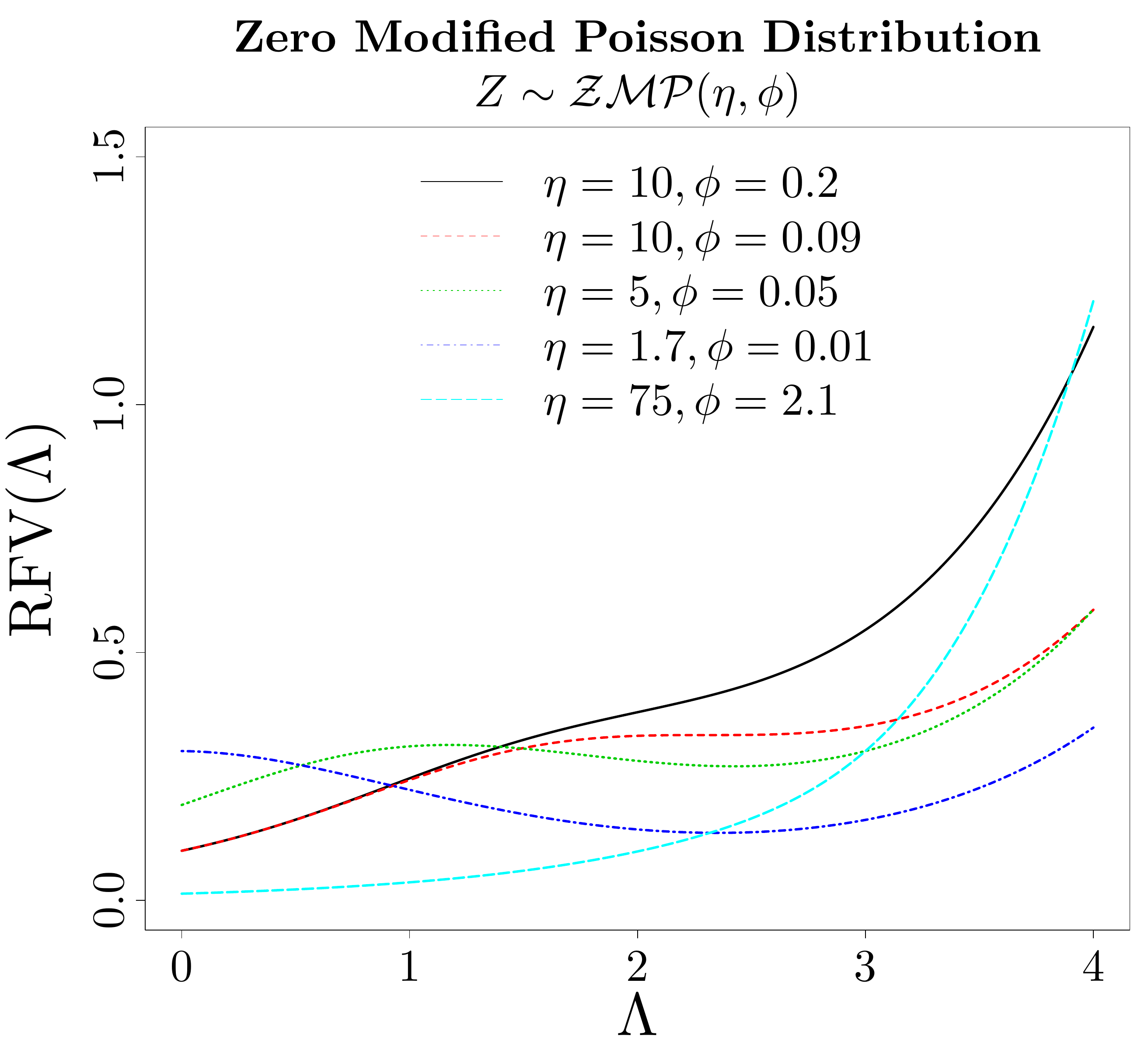}}
       \caption{}\label{fig:RFVplotvi}
\end{subfigure}
\caption{Shapes of the $\text{RFV}$ for discrete (shared) frailty models.}
\label{fig:RFVplot}
\end{figure}
The shifted models (Fig. \ref{fig:RFVplot}\subref{fig:RFVplotii},\ref{fig:RFVplot}\subref{fig:RFVplotiii},\ref{fig:RFVplot}\subref{fig:RFVplotiv}), and the $\mathcal{ZMP}$ model (Fig. \ref{fig:RFVplot}\subref{fig:RFVplotv},\ref{fig:RFVplot}\subref{fig:RFVplotvi}) have the interesting property that they are not restricted to a monotone trajectory of the $\text{RFV}$. 
The tail of the $\text{RFV}$ for the shifted models is monotonically decreasing for $p>0$, i.e. whenever $Z=0$ is not an element of the support of the frailty distribution.
A monotonically increasing tail of the $\text{RFV}$ arises only for the boundary case of $p=0$, i.e. if $Z=0$ is an element of the support of the frailty distribution.
The $\text{RFV}$ of the mentioned shifted models can have a single stationary point and therefore a non-monotone trajectory for $p>0$.


The trajectory of the $\text{RFV}$ for the $\mathcal{ZMP}$ model can be seen in Fig. \ref{fig:RFVplot}\subref{fig:RFVplotv} and \ref{fig:RFVplot}\subref{fig:RFVplotvi}.
For the $\mathcal{ZMP}$ model, the tail of the $\text{RFV}$ is monotonically increasing if $\phi>0$, i.e. if $Z=0$ is an element of the support of the frailty distribution
(Fig. \ref{fig:RFVplot}\subref{fig:RFVplotvi}).
For the boundary case of $\phi=0$, i.e. if $Z=0$ is not an element of the support of the frailty distribution, the tail of the $\text{RFV}$ is decreasing (Fig. \ref{fig:RFVplot}\subref{fig:RFVplotv}).
The $\text{RFV}$ of the $\mathcal{ZMP}$ model can have up to two stationary points.
In this case, which occurs only for $\phi>0$, a local maximum is followed by a local minimum. 
If the $\text{RFV}$ has a single stationary point, this is either a global minimum or a saddle point for $\phi>0$.
If $\phi=0$, a stationary point is the global maximum.


\begin{figure}
\centering
    \includegraphics[scale=.225]{{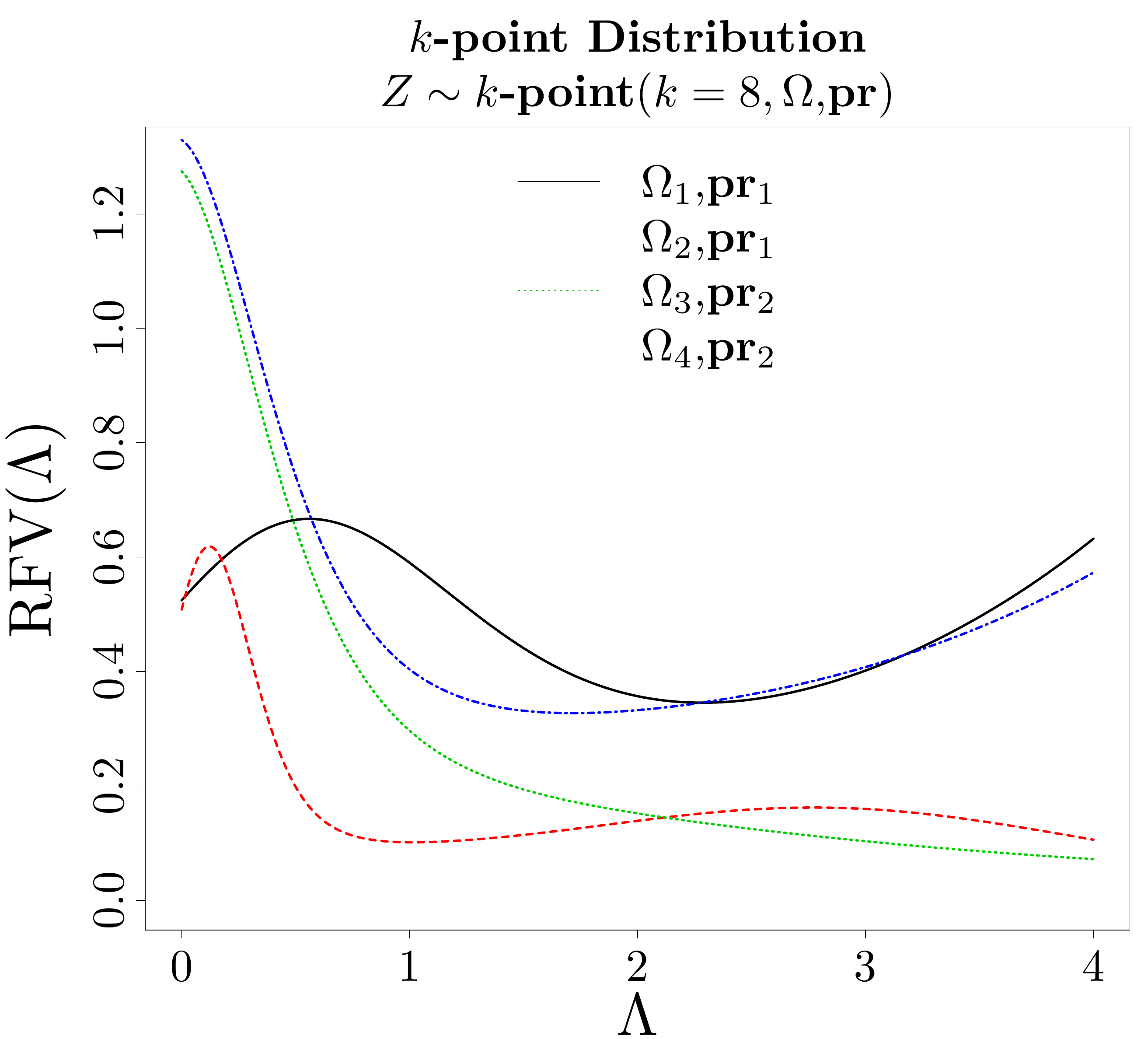}} 
\caption{Examples of shapes of the $\text{RFV}$ for $k$-point distributed frailties. The parameters are: {$\Omega_1=\{0.99,  2.02,  2.22,  2.41,  2.51,  2.52,  3.96, 10.44\}$}, $\Omega_2 = \frac{\Omega_1}{4}$ except that $z_{(1)}=0$, {$\Omega_3 = \{0.35, 0.41, 0.49, 0.60, 1.09, 1.39, 3.75, 5.63\}$}, $\Omega_4 = \Omega_3$ except that {$z_{(1)}=0$},
{\textbf{pr}$_1 = [0.03, 0.22, 0.01, 0.03, 0.18, 0.03, 0.16, 0.34]$}, { and \textbf{pr}$_2 = [0.04, 0.14, 0.25, 0.21, 0.17, 0.08, 0.07, 0.04]$}.}
\label{fig:kpoint}
\end{figure}

Figure \ref{fig:kpoint} shows some examples of the $\text{RFV}$ for the $k$-point distribution.
The parameters of the $k$-point distribution are the number of latent groups $k \in \mathbb{N}$, the support $\Omega=\{z_{(1)},\dots,z_{(k)}\} \in \mathbb{R}_{\geq 0}^k$, with $z_{(1)}<\dots<z_{(k)}$, as well as the probabilities pr$_{(m)}=P(Z=z_{(m)})$, for $m=1,\dots,k$, which we collect in the vector \textbf{pr}.
The $\text{RFV}$ of the $k$-point distribution equals $\frac{\sum_{m'=1}^k\sum_{m=1}^k z_{(m)}^2\exp\{-(z_{(m)}+z_{(m')})\Lambda\}\text{pr}_{(m)}\text{pr}_{(m')}}{\sum_{m'=1}^k\sum_{m=1}^k z_{(m)}z_{(m')}\exp\{-(z_{(m)}+z_{(m')})\Lambda\}\text{pr}_{(m)}\text{pr}_{(m')}}-1$.
Note that in contrast to Fig. \ref{fig:RFVplot}, the trajectories of the $\text{RFV}$ in Fig. \ref{fig:kpoint}, where $k$ is chosen to be $8$, are merely illustrative examples and the shapes of the $\text{RFV}$ are not exhaustively studied for the family of k-point distributions.
Moreover, note that Fig. \ref{fig:RFVplot}\subref{fig:RFVplotiii} represents a special case of the $2$-point distribution.
Monotone as well as non-monotone trajectories of the $\text{RFV}$ are observed.
Fig. \ref{fig:kpoint} shows trajectories with up to three stationary points.
For the examples in Fig. \ref{fig:kpoint}, a strictly decreasing tail arises if $z_{(1)} > 0$ and a strictly increasing tail arises if $z_{(1)}=0$.
We simulated a couple of examples of the $k$-point distribution and never found a distribution for $Z$ where the $\text{RFV}$ had more than $4$ stationary points even for very large $k$.

\subsection{Limiting behaviour of the RFV}
Note that for all models that are displayed in Fig. \ref{fig:RFVplot} {as well as the examples in Fig. \ref{fig:kpoint}}, an increasing tail of the $\text{RFV}$, and therefore of the $\text{CRF}$, arises only if the discrete frailty distribution has positive probability on $Z=0$,
whereas a decreasing tail occurs only if the discrete distribution has zero probability on $Z=0$ instead. 
This is no coincidence as the following Proposition shows.

\begin{proposition}\label{prop} Let $Z\geq0$ be a discrete RVa {for which the support can be arranged in ascending order, i.e. $\Omega = \{z_{(1)}, z_{(2)}, \dots\}$.
The support $\Omega$ is bounded from below by minimum $z_{(1)}$, but not necessarily from above.}
Then, in a time-invariant (shared) frailty model according to equation \eqref{model}, where the frailties $z$ are realizations of $Z$ and a $t_j \in \boldsymbol{t}$ exists such that $\lim_{t_j \to \infty} \Lambda_0(\boldsymbol{t}) \to \infty$, 
\begin{align*}
\lim_{t_j \to \infty} \textup{RFV}(\Lambda_0(\boldsymbol{t})) =
\begin{cases}
\infty & \text{ if }  z_{(1)} = 0  \textup{ (i)} \\
0 & \text{ if }  z_{(1)}>0 \textup{ (ii)}. 
\end{cases}
\end{align*}
\end{proposition}

\begin{proof}
The proof of Proposition \ref{prop} can be found in Appendix \ref{proof}.
\end{proof}

Hence, if the assumptions of Proposition \ref{prop} hold and the $\text{RFV}$ ($\text{CRF}$) has a tail, the tail of the $\text{RFV}$ ($\text{CRF}$) is strictly increasing if and only if the smallest frailty is zero and it is strictly decreasing if and only if the smallest frailty is larger than zero in a discrete time-invariant (shared) frailty model.

{The results of Proposition \ref{prop} can also be extended to a certain kind of discrete time-varying (shared) frailty model.
Let $Z(t)$ be piecewise constant on disjoint time intervals, i.e. $Z(t)=\sum_{q=1}^Q \mathcal{I}(t \in [t_{q-1},t_{q}))Z_q$,
with $t_{Q} = \infty > t_{Q-1} >\dots >t_1 > 0 = t_0$, $\mathcal{I}$ being the indicator function which is one if the expression in the brackets is true and zero otherwise, and each $Z_q$, with $q=1,\dots, Q$, being a discrete RVa.
Then, the support of $Z(t_Q)$ 
determines the long-run behaviour of the $\text{RFV}$ ($\text{CRF}$) according to Proposition \ref{prop}.
The corresponding proof can be thought of as stepping in at time $t_Q$.
Note that we do not claim that this is the least restrictive set of assumptions such that the results of Proposition \ref{prop} carry over to discrete time-varying frailty models, but the easiest set of assumptions to adapt the proof to the time-varying scenario.
The adaption of the proof is also outlined in Appendix \ref{proof}.}

The results of Proposition \ref{prop} can be interpreted as follows.
We start with the cluster perspective, focusing on the $\text{CRF}$, and then take the heterogeneity perspective, i.e. focusing on the $\text{RFV}$.
The interpretation of both perspectives relies on the fact that $P(Z=z_{(1)} \mid \boldsymbol{T}> \boldsymbol{t})$ increases with $\boldsymbol{t}$ and eventually approaches one.

In the case of $z_{(1)}=0$, clusters can increasingly well be separated into non-frail and frail clusters with increasing time.
In that sense, event-free clusters are more and more unlikely to experience any events in the future, as it becomes more and more likely that those clusters are not frail.
Frail clusters, however, will experience events after some timespan. 
This means that event-free clusters and partially event-free clusters up to some large $\boldsymbol{t}$, can almost be perfectly separated in those clusters, who will never experience any event and those who will eventually experience all events.
Hence, the association within a cluster in terms of the $\text{CRF}$ approaches infinity.

In the case of $z_{(1)}>0$, all clusters are frail. 
Hence, all clusters eventually experience events.
As clusters with $z_{(1)}$ eventually dominate the population,
but also experience events, the hazard of those clusters dominates the population hazard for large $\boldsymbol{t}$ in the definition of the $\text{CRF}$.
Thus, for large $\boldsymbol{t}$, information on the presence or absence of events does not supply relevant information on the hazard rate of the cluster.
Hence, association within a cluster, as represented by the $\text{CRF}$, diminishes in the long-run.

In the context of heterogeneity, i.e. focusing on the $\text{RFV}$, the expectation of $\lambda^{(j)}(t \mid Z)$ given an event-free observational unit, $h^{(j)}(t_j \mid T{(l)>t_l, l \neq j})$, approaches $z_{(1)} \lambda^{(j)}_0(t)$  for large $\boldsymbol{t}$.
The variation of observational unit-specific hazard rates, $Z \lambda^{(j)}_0(t)$, given the absence of events approaches zero,
as clusters with frailty $z_{(1)}$ dominate the population for large $\boldsymbol{t}$. 
Hence, if $z_{(1)}>0$, the variation of observational unit-specific hazard rates of survivors relative to the squared population hazard approaches zero.
In the case of $z_{(1)}=0$, however, the population hazard rate $h^{(j)}(t_j \mid T{(l)>t_l, l \neq j})$ approaches zero and the variation relative to the squared population hazard approaches infinity. 
    
Proposition \ref{prop} also shows the importance of having flexibility in the support of the discrete frailty distribution. 
Restrictions on the support determine the trajectory of the tail of the $\text{RFV}$.
In many cases, the choice of the smallest possible frailty value determines the entire trajectory of the $\text{RFV}$, as most standard discrete frailty distributions have a monotone $\text{RFV}$. 
More flexible models regarding the support might also be exploited to verify the assumption of a cure rate.
In this sense, an estimated decreasing $\text{RFV}$ can be seen as counter-evidence of the existence of a cure rate, whereas an increasing $\text{RFV}$ might be regarded as an indication of the presence of a cure rate. 

Proposition \ref{prop} also implies that in a discrete (shared) frailty setting, the $\text{RFV}$ (and hence the $\text{CRF}$) can never converge to a positive asymptote (greater than one), which might be seen as a disadvantage. 
This makes it even more important that continuous special cases are nested within an otherwise discrete family of distributions, for which the $\text{RFV}$ can converge towards a positive asymptote, such as the $\mathcal{G}$ in the $\mathcal{AF}$ for which the $\text{RFV}$ is constant.

Alternatively, the underlying assumptions of Proposition \ref{prop} might be violated such that the limiting behaviour of the $\text{RFV}$ differs from the two extreme scenarios of the (shared) frailty model as stated in Proposition \ref{prop}. We discuss two examples; firstly, a correlated frailty model and secondly, models with time-varying support.

In the multivariate context ($J>1$), discrete correlated frailty models might be utilized to obtain a model for which the $\text{CRF}$ is converging towards a positive asymptote greater than one.
To show this, we briefly consider the models introduced by \cite{Cancho.2020b}.
The conditional, target-specific hazard rate is assumed to be $\lambda^{(j)}(t \mid Z^{(j)}=z^{(j)}) = z^{(j)} \lambda^{(j)}_0(t)$, for $j=1,\dots,J$.
Given $W=w$, the frailty RVas $Z^{(j)}$, $j=1,\dots,J,$ are distributed as $\mathcal{P}(\eta_j w)$ and conditionally independent for any pair $Z^{(j)}$ and $Z^{(j')}$ and $j\neq j'$.
The RVa $W\geq 0$ is the association-inducing feature in the model, which has distribution $\pi_W$.
The correlation of $Z^{(j)}$ and $Z^{(j')}$ equals $\frac{\operatorname{Var}(W) \sqrt{\eta_{j} \eta_{j'}}}{\sqrt{(\operatorname{Var}(W)\eta_j+\operatorname{E}(W))(\operatorname{Var}(W)\eta_{j'}+\operatorname{E}(W))}}
$, for $j\neq j'$, which has values in the interval $(0,1]$ for $\operatorname{Var}(W),\eta_j,\eta_{j'}>0$.
This model is mathematically convenient, as the Laplace transform of the $\mathcal{P}$ is an exponential function (see Table \ref{table:RFVs} in Appendix \ref{ShapesApp}) and hence, $S({\boldsymbol{t} \mid W=w}) = \exp\{-w d(\boldsymbol{t})\}$, with $d(\boldsymbol{t}) = \sum_{j=1}^J \eta_j (1-\exp\{-\Lambda_0^{(j)}(t_j)\}$. 
Thus, $S(\boldsymbol{t}) = \mathcal{L}_{\pi_W}(d(\boldsymbol{t}))$, i.e. the survival function is determined by the Laplace transform of $W$ which is identical to a shared frailty model with RE distribution $W$, except for the fact that its implicit argument is determined by the $\mathcal{P}$ distributions through $d \in \{ d(\boldsymbol{t}):\boldsymbol{t} \in \mathbb{R}_{\geq 0})\}$.
Note that $\lim_{t_j \to \infty}d(\boldsymbol{t}) = \eta_{j} + \sum_{j' \neq j} \eta_{j'} (1-\exp\{-\Lambda_0^{(j')}(t_{j'})\}$ and
$\lim_{\boldsymbol{t} \to \infty} d(\boldsymbol{t}) = \sum_{j =1}^J \eta_j $.

\cite{Cancho.2020b} investigate the $\text{CRF}$ for $W \sim \mathcal{G}$ in their proposed correlated frailty model.
However, we assume an arbitrary choice of $\pi_W$.
As the cross-derivative of $d(\boldsymbol{t})$ with respect to $t_j$ and $t_{j'}$, for $j \neq j'$, equals zero, the $\text{CRF}$ is independent of the choice of the tuple $(j,j')$.
Moreover, the $\text{CRF}$ can be expressed as a function of $d$ and is solely determined through the Laplace transform of $W$, i.e.
\begin{align*}
\text{CRF}(\boldsymbol{t}) &= \text{CRF}(d(\boldsymbol{t})) \\ 
&= \frac{\mathcal{L}_{\pi_W}''(d(\boldsymbol{t})) \mathcal{L}_{\pi_W}(d(\boldsymbol{t}))}{\mathcal{L}'_{\pi_W}(d(\boldsymbol{t}))^2}.
\end{align*}
{Consequently, $W$ determines the shape of the $\text{CRF}$ over $d \in [0,\sum_{j=1}^J \eta_j]$, as if it were the RE in a shared frailty model.
However, the $\mathcal{P}$ distributions determine how quickly this shape evolves over time $\boldsymbol{t}$ through $d(\boldsymbol{t})$.
Eventually, the $\mathcal{P}$ distributions also set the limit of the $\text{CRF}$, $\lim_{\boldsymbol{t} \to \infty} \text{CRF}(d(\boldsymbol{t})) = \text{CRF}(\sum_{j=1}^J \eta_j)$, which will typically be a positive constant greater than one.}
Some examples of trajectories of the $\text{CRF}$ for a corresponding choice of $\pi_W$ can be seen in Fig. \ref{fig:RFVplot}, where only the axis of abscissas needs to be truncated at $\sum_{j=1}^J \eta_j$ and the value of one has to be added to the axis of ordinates.
Hence, within the class of discrete correlated frailty models introduced by \cite{Cancho.2020b}, the $\text{CRF}$ is forced to approach a positive constant greater than one in general, which is in contrast to discrete shared frailty models.
Note that the $\text{CRF}$ looses its connection to the $\text{RFV}$ in the case of correlated frailty models.

{ So far we only considered a subset of the models introduced by \cite{Caroni.2010}, namely the shifted models.
The shifted models fall into the jurisdiction of Proposition \ref{prop}.}
However, if we allow the shifting parameter to be time-dependent, the models do not satisfy our definition of a time-invariant (shared) frailty model according to equation \eqref{model} and hence violate the assumptions of Proposition \ref{prop}.
The discrete latent RV for such models is $Z(t) = Z_* + p(t)$, where $Z_*\geq 0$ is some discrete RV and $p(t)\geq 0$ is the time-dependent shifting parameter.
The hazard rate equals $\lambda^{(j)}({t \mid Z_*=z_*}) = z(t)\lambda^{(j)}_0(t)$, where $z(t)=z_*+p(t)$. 
The $\text{RFV}$ of this model equals $\text{RFV}(t) = \text{RFV}_*(\Lambda_0(t)) \big[ \frac{{\mathcal{L}_*}' (\Lambda_0(t))}{{\mathcal{L}_*}'(\Lambda_0(t)) - p(t)\mathcal{L}_*(\Lambda_0(t))} \big]^2$, similar to those of the shifted models. 
If $Z_*$ is chosen to be $\mathcal{P}$ distributed and $p(t)=\exp\{-\frac{\Lambda(t)}{2}\}\eta$,
then $\text{RFV}(t) = \frac{1}{\eta} \big[   \exp\{-\frac{\Lambda(t)}{2} \}+1  \big]^{-2}$ which converges to $\frac{1}{\eta}$. 
Note that $p(t)$ approaches zero in the long run is a necessary condition for the $\text{RFV}$ to approach a positive asmyptote for this class of models (irrespective of the choice of the $\mathcal{P}$ distribution), but not a sufficient condition.
If $p(t)=\exp\{-\frac{\Lambda(t)}{2}\}\eta (\sin(\Lambda(t))+2)$ instead, the $\text{RFV}(t)=\frac{1}{\eta} \big[  \exp\{-\frac{\Lambda(t)}{2} \}+\sin(\Lambda(t))+2 \big]^{-2}$, {which neither converges nor approaches infinity}. 
If $p(t) = \exp\{-\Lambda(t)\}\eta$ the $\text{RFV}$ approaches infinity.
{If $p(t)>\epsilon$ for some $\epsilon \in \mathbb{R}_{>0}$ and all $t$ greater than some time threshold, the $\text{RFV}$ will approach zero, irrespective of the choice of the $\mathcal{P}$ distribution for $Z_*$.}


\section{Discussion} \label{conc}

%
%

For time-invariant discrete univariate and shared frailty models, we have investigated the trajectory of the $\text{RFV}\ (\text{CRF})$. 
We have shown that the choice of the smallest frailty category determines the long-run behavior of both functions. 
If the frailty distribution, for which the support can be ordered from the smallest to larger values,
includes an atom at zero, such as models based on the classical $\mathcal{B}$, $\mathcal{NB}$, and $\mathcal{P}$ distributions do, the $\text{RFV}\ (\text{CRF})$ will approach infinity. 
If the discrete frailty distribution does not contain an atom at zero, the $\text{RFV}\ (\text{CRF})$ will approach zero (one).

Families of discrete distributions showing more flexibility with respect to the functional form of the $\text{RFV}$, that is, discrete families including types of distributions with and without an atom at zero, seem to be more adequate as a general modelling approach and may serve as a diagnostic tool for finding evidence for the existence of a cure rate for the data at hand. 
The Addams family, the zero-modified power series distribution as well as shifted distributions fulfill that criterion.

Among the distributions that are discussed in this paper, the Addams family introduced by \cite{Farrington.2012} has the advantage that neither a cure rate nor its absence is on the edge of the parameter space. 
So far no estimation procedures are existent for the Addams family, however.
Alternatively, zero modified power series distributions might be deployed, which manipulate the cure rate by means of an additional parameter. 
The absence of a cure rate, however, is on the edge of the parameter space of the deflation and inflation parameter, which might lead to difficulties in the estimation process for data where a cure rate seems inappropriate.
As illustrated in Section \ref{Shapes}, a straightforward way to achieve flexibility in the trajectory of the $\text{RFV}$ is to shift the support of the frailty distribution by means of an additional parameter, which varies between zero and infinity.
In that case, however, a cure rate is on the edge of the parameter space of the shifting parameter. 

Such shifted and zero-modified models have the interesting property that the $\text{RFV}$ is not restricted to be monotone.
This is a rare property for time-invariant frailty distributions, as, for example, the power variance family, the extended generalized gamma family, and the Kummer family lead to a monotone trajectory of the $\text{RFV}\ (\text{CRF})$.
A discussion of the three families, accompanied by an investigation of their $\text{RFV}$ can be found in \cite{Farrington.2012}.  

We have shown that for time-invariant discrete univariate and shared frailty distributions, the $\text{RFV}\ (\text{CRF})$ is not able to converge to a positive asymptote (greater than one). 
This is in contrast to existing (families of) continuous distributions, e.g. the Kummer family. 
The Addams family of discrete frailty distributions circumvents this issue by including the gamma distribution as a continuous special case that has a constant $\text{RFV}$ ($\text{CRF}$) and hence a positive asymptote (greater than one), unless the frailty variance is zero.
This case might be approached in the estimation process if the non-zero, but finite $\text{RFV}$ in the long run is important for modelling the data at hand.

A positive asymptote of the $\text{RFV}$ might be approached if we leave the framework of a time-invariant univariate or shared frailty model.
If the support-shifting parameter of the { discrete frailty} models introduced by \cite{Caroni.2010} is allowed to be time dependent, a positive asymptote {of the} $\text{RFV}$ might be approached if the support-shifting parameter converges towards zero in the long run.
In the multivariate setting with discrete frailty, the $\text{CRF}$ approaches a positive constant greater than one { for the class of discrete correlated frailty models introduced by \cite{Cancho.2020b}.
The $\text{RFV}$ and the $\text{CRF}$ loose its connection in that case}.

Our focus was on discrete frailty models in which the random effect is time-invariant.
We have extended our results to piecewise constant discrete time-varying shared frailty models, however. 
For this class of frailty models, the $\text{RFV}$ ($\text{CRF}$) approaches either infinity or zero (one), depending on whether the final random effect includes zero in its support or not.
The shifted models inspired by \cite{Caroni.2010} discussed in Section \ref{Frailty} might be exploited to generate more flexible, non-monotone trajectories of the $\text{RFV}\ (\text{CRF})$, by leaving the setting of a time-invariant univariate and shared frailty model and letting the support-shifting parameter of the frailty distribution be piecewise constant on disjoint time intervals.
This would keep the model structure relatively simple and might lead to several switches of periods with decreasing and increasing $\text{RFV}$ ($\text{CRF}$). 
Additionally, {this model is} a special case of the aforementioned piecewise constant time-varying frailty models and hence the $\text{RFV}$ ($\text{CRF}$) either approaches zero (one) or infinity depending on whether the last shift is greater than zero or equal to zero, respectively.

The $\text{RFV}$ of discrete frailty models has shown more diverse profiles compared to existing standard continuous frailty distributions.
This { property} might be exploited to develop time-varying correlated discrete frailty models with flexible trajectories of the $\text{RFV}$ ($\text{CRF}$).



\backmatter

\bmhead{Acknowledgments}
This research has received funding from the Deutsche Forschungsgemeinschaft
(DFG; German Research Foundation, grant number UN 400/2-1).
We thank Dr. Fabian K\"uck for comments on the paper.



\begin{appendices}

\section{The redefined CRF and its connection to the RFV} \label{CRFtoRFV}
 
We  show that $\text{CRF}(\boldsymbol{t})=1+\text{RFV}(\boldsymbol{t})$ for $J\geq2$.
Let $S(\boldsymbol{t})$ denote the population survival function, i.e. $S(\boldsymbol{t}) = P(\boldsymbol{T}>\boldsymbol{t})$ (given covariates).
Note, that we utilize Lagrange's notation of derivatives for functions with a single argument (e.g. $\mathcal{L}$),
but Leibniz's notation for functions with multiple arguments (e.g. $S, \Lambda$).
The population hazard of $T^{(j)}$ given $T^{(j')}=t_{j'}$, $j' \neq j$, is
\begin{align*}
h^{(j)}(t_j \mid T^{(j')}=t_{j'}, T^{(l)} >t_l\ \forall l \neq  j,j') 
&= -\frac{\partial^2 S(\boldsymbol{t}) / \partial t_j \partial t_{j'}}{\partial S(\boldsymbol{t}) / \partial t_{j'}}.  
\end{align*}
Re-expressing this in terms of the Laplace transform yields
\begin{align*}
h^{(j)}(t_{j} \mid T^{(j')}=t_{j'},  T^{(l)} >t_l\ \forall l \neq  j,j') &= 
-\frac{\mathcal{L}''(\Lambda_0(\boldsymbol{t})) }{\mathcal{L}'({\Lambda}_0(\boldsymbol{t}))} \frac{\partial {\Lambda}_0(\boldsymbol{t}) / \partial t_j \ \partial {\Lambda}_0(\boldsymbol{t}) / \partial t_{j'} }{\partial {\Lambda}_0(\boldsymbol{t}) / \partial t_{j'}} \\
&=  -\frac{\mathcal{L}''({\Lambda}_0(\boldsymbol{t}))}{\mathcal{L}'({\Lambda}_0(\boldsymbol{t}))} \partial {\Lambda}_0(\boldsymbol{t}) / \partial t_j.  
\end{align*}
Proceeding analogously with the denominator of the $\text{CRF}$ yields
\begin{align*}
h^{(j)}(t_{j} \mid   T^{(l)} >t_l\ \forall l \neq  j) &= -\frac{\partial S(\boldsymbol{t}) / \partial t_j}{S(\boldsymbol{t})} \\
&= -\frac{\mathcal{L}'(\Lambda_0(\boldsymbol{t}))}{\mathcal{L}(\Lambda_0(\boldsymbol{t}))} \partial \Lambda_0(\boldsymbol{t}) / \partial t_j
\end{align*}

This leads to $\text{CRF}^{(j,j')}(\boldsymbol{t}) = \frac{\mathcal{L}''(\Lambda_0(\boldsymbol{t})) \mathcal{L}(\Lambda_0(\boldsymbol{t}))}{\mathcal{L}'(\Lambda_0(\boldsymbol{t}))^2}$.
Note, that the baseline hazard rates cancel out and hence the choice of $j$ and $j'$ is irrelevant for the $\text{CRF}$ in the case of a time-invariant shared frailty model and the superscripts might be dropped accordingly. Furthermore, the $\text{CRF}$ might be seen as a function of the sum of the cumulative baseline hazards, i.e. $\text{CRF}(\boldsymbol{t})= \text{CRF}(\Lambda_0(\boldsymbol{t}))$.

The $\text{RFV}$ can be expressed through the conditional moment generating function and hence, through the conditional Laplace transform, i.e. the Laplace of survivors.
The Laplace of survivors 
$\mathcal{L}_{\boldsymbol{t}}(s) = \int_0^\infty\exp\{-zs\}g(z \mid \boldsymbol{T}>\boldsymbol{t})dz = \frac{\mathcal{L}(s + \Lambda_0(\boldsymbol{t}))}{\mathcal{L}( \Lambda_0(\boldsymbol{t}))}$
can be utilized to obtain $\operatorname{E}(Z \mid \boldsymbol{T}>\boldsymbol{t})$ through $- \mathcal{L}'_{\boldsymbol{t}}(0)$ and $\operatorname{E}(Z^2 \mid \boldsymbol{T}>\boldsymbol{t})$ through $ \mathcal{L}''_{\boldsymbol{t}}(0)$.
Thus,
\begin{align*}
\text{RFV}(\boldsymbol{t}) &= \frac{\operatorname{Var}(Z \mid \boldsymbol{T}>\boldsymbol{t})}{\operatorname{E}(Z \mid \boldsymbol{T}>\boldsymbol{t})^2} \\
&= \frac{\mathcal{L}''(\Lambda_0(\boldsymbol{t})) \mathcal{L}(\Lambda_0(\boldsymbol{t})) - \mathcal{L}'(\Lambda_0(\boldsymbol{t}))^2}{\mathcal{L}'(\Lambda_0(\boldsymbol{t}))^2} \\
&= \text{CRF}(\boldsymbol{t}) - 1.
\end{align*}
Analogously to the $\text{CRF}$, the $\text{RFV}$ might be seen as a function of the sum of the cumulative baseline hazards.

\section{Shape of the RFV of selected models} \label{ShapesApp}

In the $\mathcal{ZMP}$ model, the derivative of the $\text{RFV}$ with respect to the generic time $\Lambda$ is
\begin{align*} 
\text{RFV}'(\Lambda) = \text{RFV}_{\mathcal{P}}(\Lambda) \bigg( 1 + \frac{b}{r(\Lambda)} \bigg),
\end{align*} 
with $b = \frac{\phi -1}{1- \phi \exp\{-\eta\}}$ and $r(\Lambda) = \frac{\text{RFV}_{\mathcal{P}}(\Lambda)^2}{(1+\text{RFV}_{\mathcal{P}}(\Lambda))^2 - \text{RFV}_{\mathcal{P}}(\Lambda)} \exp\{\frac{1}{\text{RFV}_{\mathcal{P}}(\Lambda)}\}$. 
For $\phi\geq1$ (inflation) the $\text{RFV}$ is monotonically increasing, whereas stationary points might exist if $\phi <1$ (deflation).
A stationary point satisfies $-b = r(\Lambda)$.
In particular, $r(0) = \frac{\exp\{\eta\}}{{\eta}^2 + \eta + 1}$ and 
$\lim_{\Lambda \to \infty}r(\Lambda) = 1$ and $r(\Lambda)$ is strictly increasing for $\eta \leq 1$ but has a global minimum in the interval $(0,\infty)$ for $\eta>1$, that is at $\Lambda$ for which $\text{RFV}_{\mathcal{P}}(\Lambda) = 1$, and takes the value $\frac{\exp\{1\}}{3}$. 
Hence, there might be none, one or two stationary points in the $\text{RFV}$ for $\phi<1$.
Note that in case of $\phi<1, 0<-b\leq 1$.

The second order derivative of the $\text{RFV}$ evaluated at a stationary point is 
\begin{align*}
\frac{\text{RFV}_{\mathcal{P}}(\Lambda)-1}{(1+\text{RFV}_{\mathcal{P}}(\Lambda))^2-\text{RFV}_{\mathcal{P}}(\Lambda)}.
\end{align*}
Thus, if the stationary point is at $\Lambda$ for which $\text{RFV}_{\mathcal{P}}(\Lambda) = 1$, the $\text{RFV}$ has a single stationary point which is a saddle point and the \text{RFV} is monotonically increasing.
In any other case, a stationary point is either a local minimum or a local maximum.
If the $\text{RFV}$ has two stationary points, a local maximum is followed by a local minimum. 
For $\phi>0$, the $\text{RFV}$ is strictly increasing in its tail.
A strictly decreasing tail of the $\text{RFV}$ exists if and only if $\phi=0$, i.e. for the truncated $\mathcal{P}$ model.

Table \ref{table:RFVs} gives details of the RE models discussed in Section \ref{Shapes}.
For the shifted models that are listed in Table \ref{table:RFVs}, we assume that $p>0$.
The corresponding information for $p=0$ can be obtained from the row of the non-shifted counterpart.
For the $\mathcal{NB}+p$ model, $c(\Lambda)$ is defined to be $\nu(1-\pi) + p(\exp\{\Lambda\} - (1-\pi))$.
The constant $c^*$ is equal to $\ln\{\frac{(1-\pi)(\nu - p)}{p}\}$ for $\nu > p > 0$ and $-\infty \text{ for } \nu \leq p$.
For the $\mathcal{B}+p$ model, $c(\Lambda)$ is defined to be ${\pi \exp\{-\Lambda\}(p+n) + p(1-\pi)}$,
and the constant $c^*$ is equal to $-\ln\{\frac{p(1-\pi)}{\pi(p+n)}\}$.
For the $\mathcal{P}+p$ model, $c(\Lambda)$ is defined to be $\eta \exp\{-\Lambda\} + p$, the constant $c^*$ equals ${\ln\{\frac{\eta}{p}\}}$.

\begin{sidewaystable*}
\centering
\scalebox{0.885}{
\begin{tabular}[t]{ |p{2cm}|p{3.75cm}|p{3cm}|p{8cm}|p{3.25cm}| } 
\hline
{$Z \sim$} & $\mathcal{L}(s)$ & $\text{RFV}(\Lambda)$ & $\text{RFV}'(\Lambda)$ & Support  \\[1ex]
\hline
\hline
 $\mathcal{NB}(\pi,\nu)$  & $\frac{\pi}{1-(1-\pi)\exp\{-s\}}^\nu$ 
& ${\exp\{\Lambda\}} / {(1-\pi)\nu}$ &  $ {>0 \forall \Lambda}$ & $ {\{0,1,\dots\}}$ \\[1ex]
\hline
$\mathcal{NB}_{>0}(\pi,\nu)$ & $\frac{\pi}{\exp\{s\}-(1-\pi)}^\nu$  
& ${(1-\pi)\exp\{-\Lambda\}} / {\nu}$ & $ {< 0 \forall \Lambda}$ & $ {\{\nu,\nu+1,\dots\}}$  \\[1ex]
\hline
$\mathcal{B}(\pi,n)$  &$ {((1-\pi)+\pi \exp\{-s\})^n}$ 
& ${(1-\pi)\exp\{\Lambda\}} / {n\pi}$ &  $ {>0 \forall \Lambda}$ &$ { \{0,\dots,n\}}$ \\[1ex]
\hline
$\mathcal{P}(\eta)$  & $  {\exp\{\eta (\exp\{-s\}-1)\}}$
 & ${\exp\{\Lambda\}} / {\eta}$ &  $ {>0 \forall \Lambda}$ & $ {\{0,1,\dots\}}$  \\[1ex]
\hline
$\mathcal{NB}(\pi,\nu) + p$ & 
\begin{tabular}{@{}l@{}l@{}}
			$ {\mathcal{L}_{\mathcal{NB}}(s)}$\\
                   $ {\times \exp\{-s p\}}$\\
                 \end{tabular}
& \begin{tabular}{@{}l@{}l@{}}
                   ${\exp\{\Lambda\}\nu(1-\pi)}$\\
                   $ \times {c(\Lambda)^{-2}}$\\
                 \end{tabular}  &
\begin{tabular}{@{}l@{}l@{}}
                   $ {>0 \ \forall \Lambda: \Lambda < c^*}$,\\
                   $ {=0 \ \forall \Lambda: \Lambda = c^*}$,\\
			$ {<0 \text{ else}}$\\
                 \end{tabular}  & 
$ {\{p,p+1,\dots\}}$ \\[1ex]
\hline
$\mathcal{B}(\pi,n) + p$ & 
\begin{tabular}{@{}l@{}l@{}}
			$ {\mathcal{L}_{\mathcal{B}}(s)}$\\
                   $ {\times \exp\{-s p\}}$\\
                 \end{tabular}
& \begin{tabular}{@{}l@{}l@{}}
                   ${(1-\pi)\pi n \exp\{-\Lambda\}}$\\
                   $ \times {c(\Lambda)^{-2}}$\\
                 \end{tabular}&
\begin{tabular}{@{}l@{}l@{}}
                   $ {>0  \ \forall \Lambda: \Lambda < c^*,}$\\
                   $ {=0 \ \forall \Lambda: \Lambda = c^*,}$\\
			$ {<0 \text{ else}}$\\
                 \end{tabular}  & 
$ {\{p,p+1,\dots\,p+n\}}$\\[1ex]
\hline

$\mathcal{P}(\eta) + p$ & 
\begin{tabular}{@{}l@{}l@{}}
			$ {\mathcal{L}_{\mathcal{P}}(s)}$\\
                   $ {\times \exp\{-s p\}}$\\
                 \end{tabular}
& ${\eta \exp\{-\Lambda\}} / {c(\Lambda)^2}$  &
\begin{tabular}{@{}l@{}l@{}}
                   $ {>0  \ \forall \Lambda: \Lambda < c^*}$,\\
                   $ {=0 \ \forall \Lambda: \Lambda = c^*,}$\\
			$ {<0 \text{ else}}$\\
                 \end{tabular} & 
$ {\{p,p+1,\dots\}}$ \\[1ex]
\hline

 \multirow{4}{*}{$\mathcal{AF}(\alpha,\gamma)$} &
			 \multirow{4}{*}{ 
							{\begin{tabular}[t]{p{5cm}}See \\ \cite{Farrington.2012} \end{tabular}}
}& 
	\multirow{4}{*}{ $ {\gamma\exp\{\alpha \Lambda\}}$} &
                   {\begin{tabular}[t]{p{2cm}p{3cm}}For ${\alpha < 0: }$ & ${<0 \forall \Lambda}$ \end{tabular}} & $ {\{-\alpha \nu, -\alpha(1 + \nu), \dots\}}$ \\    [1ex]         
			\cdashline{4-5}
			& & &{\begin{tabular}[t]{p{2cm}p{3cm}}For ${\alpha = 0: }$ & ${=0 \forall \Lambda}$ \end{tabular}}  & $ {\mathbb{R}_{>0}}$\\[1ex]
                   \cdashline{4-5}
		     & & &{\begin{tabular}[t]{p{2cm}p{3cm}}For ${\alpha > \gamma: }$ & ${>0 \forall \Lambda}$ \end{tabular}}&  $ {\{0, \alpha, 2\alpha,\dots\}}$ \\[1ex]
                   \cdashline{4-5}
		    &   & &{\begin{tabular}[t]{p{4.5cm}p{3cm}}For ${0 <\alpha < \gamma, \frac{1}{\gamma-\alpha} \in \mathbb{N}: }$ & ${>0 \forall \Lambda}$ \end{tabular}} & $ {\{0, \alpha, 2\alpha,\dots,b\alpha\}}$ \\[1ex]
\hline

 \multirow{4}{*}{$\mathcal{ZMP}(\eta, \phi)$} &
			 \multirow{4}{*}{\begin{tabular}{@{}l@{}l@{}}
				$  {\frac{1-\phi \exp\{-\eta\}}{1-\exp\{-\eta\}}\mathcal{L}_{\mathcal{P}}(s)}$\\
				$  {+ \frac{\phi - 1}{\exp\{\eta\}-1}}$\\ 
               	 \end{tabular}} & 
	\multirow{4}{*}{ \begin{tabular}{@{}l@{}l@{}}
			   $ {\exp\{\frac{-1}{\text{RFV}_{\mathcal{P}}(\Lambda)}\}}$\\
			   $ {\times \frac{\phi - 1}{1-\phi \exp\{-\eta\}}}$\\
			   $ {\times (1+\text{RFV}_{\mathcal{P}}(\Lambda))}$\\
			   $ {+ \text{RFV}_{\mathcal{P}}(\Lambda)}$\\
	 		\end{tabular}} &
			\multirow{4}{*}{ \begin{tabular}{@{}l@{}l@{}} {\begin{tabular}[t]{p{2cm}p{7.2cm}}
			   For $\phi=0:$ &  {\text{RFV} can be strictly decreasing}\end{tabular}}\\
			    {or can have inner global maximum. }\\[1ex]
			  {\begin{tabular}[t]{p{2cm}p{7.2cm}} For $\phi>0:$& { \text{RFV} can be strictly increasing}\end{tabular}}\\
			    {or can have up to two stationary points.}\\[1ex]
	 		\end{tabular}}&   \\
 		& & &   &$ {\{1,2,\dots\}}$  \\[1ex]
                   \cdashline{4-5}
 & & &    &\\[1ex]
 & & & & $ {\{0, 1,\dots\}}$\\[1ex]
\hline

\end{tabular}
}
\caption{Overview of some characteristics of discrete frailty distributions}
\label{table:RFVs}
\end{sidewaystable*}

\section{Proof of proposition \ref{prop}}  \label{proof}

Let $\Lambda \in \{\sum_{j=1}^J \Lambda^{(j)}_0(t_j): \boldsymbol{t} \in \mathbb{R}^J_{\geq 0}\}$.
The proof will be based on letting $\Lambda$ approach infinity. 
Note, that we implicitly refer to $t_j \to \infty$ in that case. 
The following proof is based on the assumption that $\operatorname{E}(Z)=\mu_1$ and $\operatorname{E}(Z^2)=\mu_2$ exist. 
This, however, does not need to be the case.
Notes on the adaption of the proof where the first second moment do not exit (at $\boldsymbol{t}=0$) can be found right after the proof in this Section.

In order to prove proposition \ref{prop} we need Tannery's Theorem.
\paragraph{Tannery's Theorem}
Let $\lim_{\Lambda^* \to \infty} v_k(\Lambda^*) = w_k$ and $F(\Lambda^*) = \sum_{k=1}^{\infty} v_k(\Lambda^*)$ be a convergent series. Then, $\lim_{\Lambda^* \to \infty} F(\Lambda^*) = \sum_{k=1}^{\infty} w_k$, if $ \mid v_k(\Lambda^*) \mid  \leq M_k$, $M_k$ being independent of $\Lambda^*$, and $\sum_{k=1}^\infty M_k < \infty$.

Tannery's Theorem can be found in \cite{Hofbauer.2002}.
Later, $\Lambda$ will take the role of $\Lambda^*$. The function $F$ will either be a pmf, a weighted survival function or an expectation and so will either be a convergent series or a finite sum, depending on $\Omega$. 
The proof of proposition \ref{prop} follows.

\begin{proof} 
\cite{Farrington.2012} proved that the \text{RFV} approaches infinity if the frailty is a mixture distribution with positive probability on $Z=0$ (p.685).
This also includes the case where the distribution for $Z>0$ is discrete and (i) follows. We prove (ii).

The conditional pmf for some $z \in \Omega$,
\begin{align}
g(z,\Lambda_0(\boldsymbol{t})) &\equiv P(Z=z \mid \boldsymbol{T} \geq \boldsymbol{t}) \nonumber \\
&= \frac{P(\boldsymbol{T} \geq \boldsymbol{t}  \mid  Z=z) g(z)}{P(\boldsymbol{T} \geq \boldsymbol{t})} \nonumber \\
&= \frac{\exp\{-z \Lambda_0(\boldsymbol{t})\} g (z)}{\sum_{k} \exp\{- z_{(k)}\Lambda_0(\boldsymbol{t})\}g(z_{(k)})} \label{gztpre} \\
&= \frac{g (z)}{\sum_{k} \exp\{- [z_{(k)}- z] \Lambda_0(\boldsymbol{t})\}g(z_{(k)})},  \label{gzT}
\end{align}
where the sum iterates over (the subscripts of the elements of) $\Omega$.

The proof is based on the idea that $\lim_{\Lambda \to \infty} g(z_{(1)},\Lambda) =  1$ and $\lim_{\Lambda \to \infty} \mu_q(t) = z_{(1)}^q$, with $ \mu_q(\boldsymbol{t}) = \operatorname{E}(Z^{q} \mid \boldsymbol{T} \geq \boldsymbol{t}), q = 1,2$. This is technically derived by Tannery's Theorem, which states the conditions to exchange the order of the limits in an expression like $\lim_{\Lambda \to \infty} \sum_{k=1}^{\infty}$. 

We will make use of Tannery's theorem two times. 
Firstly, $\lim_{\Lambda \to \infty}  g(z_{(1)},\Lambda) = 1$ follows immediately from \eqref{gzT}  and $M_k = g(z_{(k)})$ for the sum in the denominator. 
Consequently, $\lim_{\Lambda \to \infty}  g(z_{(k)},\Lambda) = 0, \forall k > 1$.
Secondly, the theorem is applied to $\operatorname{E}(z^q \mid T\geq t), q=1,2$, which requires a little more work.

For this purpose, it will be shown that a $k_{\Lambda}^*$ exists for which $\frac{\partial g(z_{(k)},\Lambda)}{\partial \Lambda} = g'(z_{(k)},\Lambda) < 0, \forall k \geq k^*_{\Lambda},$ and once a derivative becomes negative for some $\Lambda^*$ it will always be for all $\Lambda > \Lambda^*$. 
Then, $M_k = z_{(k)}^q g(z_{(k)})$ for $k\geq k_0^*$ and $k_0^*$ being the first k for which $g'(z_{(k)},\Lambda) \mid _{\Lambda=0} < 0$. 
Note, that we are referring to the right-hand side derivative for the latter case as $\Lambda \in [0,\infty)$.

The derivative of the conditional pmf w.r.t. $\Lambda$ for some $z \in \Omega$
\begin{align}
g'(z,\Lambda) &= \frac{g(z,\Lambda)^2}{g(z)} \bigg ( \sum_{k} [z_{(k)}- z]  \exp\{- [z_{(k)}- z] \Lambda\}g(z_{(k)}) \bigg ). \label{firstderiv}
\end{align}
Note, that this also requires an interchange of infinite summation and the derivation operator in the denominator of \eqref{gztpre}.
This might be shown by demonstrating uniform convergence of the series $\sum_k -z_{(k)} \exp\{-z_{(k)} \Lambda\}g(z_{(k)})$ which can be established by utilizing $z_{(k)} \exp\{-z_{(k)} \Lambda\}g(z_{(k)}) \leq z_{(k)}g(z_{(k)}) \forall \Lambda$ and the existence of the first moment $\mu_1$.
For $z_{(1)}$ the derivative includes only non-negative and at least one positive contribution as $\operatorname{Var}(Z)>0$ by assumption. 
Hence, $g'(z_{(1)},\Lambda)>0\ \forall \Lambda$. Let $k^*_{\Lambda}$ denote the index of the first negative derivative. 
Such a $k^*_{\Lambda}$ has to exist as the first derivative is positive (for all $\Lambda$) and all probabilities sum up to one w.r.t. summation over $\Omega$ (given $\Lambda$), i.e. there is at least one $k$ for which the probability decreases with time increasing. We drop the index from $k^*_{\Lambda}$ but emphasize that it depends on (generic) time. By re-expressing the derivative so that $z_{(k^*)}$ is a point of reference, i.e. 
\begin{align*}
g'(z_{(\tilde{k})},\Lambda) = &\frac{g(z_{(\tilde{k})},\Lambda)^2}{g(z_{(\tilde{k})})}\times \\
&\bigg (\exp\{[z_{(\tilde{k})} - z_{(k^*)}] \Lambda \} \sum_{k} [z_{(k)} - z_{(\tilde{k})}]  \exp\{- [z_{(k)} -  z_{(k^*)}] \Lambda\} g(z_{(k)}) \bigg ),
\end{align*}
it can be seen that $g'(z_{(\tilde{k})},\Lambda)<0, \forall \tilde{k} > k^*$. This is because the sign determining contribution within the sum, $ z_{(k)} - z_{(\tilde{k})}$, is smaller than its counterpart from $g'(z_{(k^*)},\Lambda)$, $z_{(k)} - z_{(k^*)}$, for all $k$ and the weights within the sum remain the same w.r.t. the reference $k^*$.

From \eqref{firstderiv} it can be seen that $\frac{g'(z_{(k)}, \Lambda_{+})}{g(z_{(k)}, \Lambda_{+})^2} < \frac{g'(z_{(k)},\Lambda)}{g(z_{(k)},\Lambda)^2} $, for all $\Lambda_{+} > \Lambda$ and $k>1$ as the weight for each negative contribution in the sum of the derivative increases but decreases for each positive contribution. 
That means, that once a derivative becomes negative it will stay negative with time increasing, i.e. $k_0^*\geq k_{\Lambda}^*$ for all $\Lambda>0$. As a consequence, for all $k \geq k_0^*$ and $\Lambda>0$: $g(z_{(k)}) > g(z_{(k)},\Lambda)$. 
Hence, $\mu_q(\boldsymbol{t}) \leq \mu_q$ and $\mu_q(\boldsymbol{t})$ exists (as the monotonically increasing series $\operatorname{E}(Z^q \mid \boldsymbol{T}\geq \boldsymbol{t})$ is bounded between zero and $\mu_q$) for $q=1,2$. 
Now, $M_k = z_{(k)}^{q} g(z_{(k)}) \geq z_{(k)}^{q} g(z_{(k)}, \Lambda) $ for all $k\geq k_0^{*}$, which is independent of $\Lambda$. Let $M_k = z_{(k)}^q \geq  z_{(k)}^q g(z_{(k)},\Lambda) $ for all $k< k^*_0$. 
Then, $\sum_{k=1}^{\infty} M_k = \operatorname{E}(Z^q) + \sum_{k = 1}^{k^*_0 -1} (1-g(z_{(k)}))z_{(k)}^q$ which converges as the latter part is a finite sum and finally, the sufficient conditions for exchanging the limits in $\lim_{\Lambda \to \infty} \operatorname{E}(Z^q \mid  T \geq t)$ are derived.

Through the latter application of Tannery's Theorem $\lim_{\Lambda \to \infty} \operatorname{E}(Z^q \mid \boldsymbol{T} \geq \boldsymbol{t}) = z_{(1)}^q$ follows in combination with the first application of Tannery's Theorem to $g(z_{(1)} \mid \boldsymbol{T} \geq \boldsymbol{t})$ from above. As $z_{(1)}^q>0$ by assumption, 
\begin{align*}
\lim_{\Lambda \to \infty} \text{RFV}(\Lambda) &=\lim_{\Lambda \to \infty}  \frac{\operatorname{E}(Z^2  \mid  \boldsymbol{T} \geq t) - \operatorname{E}(Z  \mid  \boldsymbol{T} \geq t)^2}{\operatorname{E}(Z  \mid  \boldsymbol{T} \geq t)^2} \\
&= \frac{z_{(1)}^2 - z_{(1)}^2}{z_{(1)}^2} = 0
\end{align*}
and the proof is complete.
\end{proof}

\paragraph{Notes on the adaption of the proof if $\operatorname{E}(Z), \operatorname{V}(Z)$ do not exist:}
In the case of $\operatorname{E}(Z) = \infty$ and/or $\operatorname{V}(Z) = \infty$ the proof can be thought to step in at some moment after $\boldsymbol{t}=0$ where $\operatorname{E}(Z\mid\boldsymbol{T} \geq \boldsymbol{t}) \equiv \mu_1(\boldsymbol{t})$ and $\operatorname{E}(Z^2\mid \boldsymbol{T} \geq \boldsymbol{t}) \equiv \mu_2(\boldsymbol{t})$ do exist.
We will show below that there is a $\boldsymbol{t}^*$ where the conditional first and second moments exist.
The proof above can be thought to step in at the moment $\boldsymbol{t}^*$ and all $\boldsymbol{t}$ in the proof can be thought to be larger than $\boldsymbol{t}^*$. 
The threshold $k_0^*$ has to be redefined to $k_{\Lambda(\boldsymbol{t}^*)}^*$.
The unconditional moments $\mu_1$ and $\mu_2$ have to be replaced by $\mu_1(\boldsymbol{t}^*)$ and $\mu_2(\boldsymbol{t}^*)$. 
Furthermore, the time argument $\Lambda$ has to be redefined to $\Lambda(\boldsymbol{t}^*) + \Lambda$ and $\Lambda \in \{\sum_{j=1}^J \Lambda^{(j)}(t_j)-\Lambda^{(j)}(t^*_j): \boldsymbol{t}\geq \boldsymbol{t}^*\}$.
Note that if $g(0)>0$, $g(0 \mid\boldsymbol{T} \geq \boldsymbol{t})>0$ and the same rule holds for $g(0)=0$.
Hence, $z_{(1)}$ decides the limiting behaviour of the $\text{RFV}$ as in the case with existing moments.

If the moments of ${Z\mid \boldsymbol{T}\geq \boldsymbol{t}}$ exist, the Laplace transform of survivors $\mathcal{L}_{\boldsymbol{t}}(s)=\frac{\mathcal{L}(s+\Lambda(\boldsymbol{t}))}{\mathcal{L}(\Lambda(\boldsymbol{t}))}$ can be utilized for their derivation by taking the derivative w.r.t. $s$ and then setting $s=0$.
We consider the numerator to proof the existence of the moments from $\boldsymbol{t}^*$ latest onwards. 

The Laplace transform $\mathcal{L}(s+\Lambda(\boldsymbol{t}))$ converges on $s \in (-\Lambda(\boldsymbol{t}),\infty)$. 
However, we are only interested in the derivative at $s=0$ (and possibly on its value for $s>0$ to compute the survival function of survivors).
Hence, we restrict the domain to $s \in (-r,\infty)$ for some fix $r \in (0,1-\frac{2}{\exp\{1\}})$.
We choose $\boldsymbol{t}^*$ such that $\Lambda(\boldsymbol{t}^*)=1$.
With that we establish uniform convergence of the sequence of partial sums $\{S_n'(s\mid \boldsymbol{t})\}_{n=1}^\infty = \sum_{k=1}^n (- z_{(k)})^q \exp\{-z_{(k)}(s+\Lambda(\boldsymbol{t}))\}g(z_{(k)})$, 
i.e. of the sum of the first ($q=1$) and second order ($q=2$) derivatives of each summand in $\mathcal{L}(s+\Lambda(\boldsymbol{t}))$ w.r.t. $s$, and thereby show that the series' limit is also the derivative of $\mathcal{L}(s+\Lambda(\boldsymbol{t}))$ for each $s\in (-r,\infty)$ given each $\boldsymbol{t} \geq \boldsymbol{t}^*$.
Note that $z^q_{(k)} \exp\{-z_{(k)}(s+\Lambda(\boldsymbol{t}))\}g(z_{(k)}) \leq z^q_{(k)} \exp\{-z_{(k)}(1-r)\}g(z_{(k)})$ for all $k$ and $s \in (-r,\infty)$ given each $\boldsymbol{t} \geq \boldsymbol{t}^*$.
The sequence $\{S_n\}_{n=1}^{\infty} = \{\sum_{k=n}^\infty (-z_{(k)})^q \exp\{-z_{(k)} (1-r) \}g(z_{(k)})  \}_{n=1}^{\infty}$
converges to $0$ as $z_{(k)}^q < \exp\{z_{(k)}(1-r)\}$ for all $k$, and $q \in \{1,2\}$ ($r$ and $\Lambda(\boldsymbol{t}^*)$ were chosen such that those inequalities holds).
Hence, for each $\epsilon > 0$ there exists a  $\tilde{n}\in\mathbb{N}$ 
such that $\lvert S_n \rvert <\epsilon, \forall n>\tilde{n}$. 
This essentially proves uniform convergence of the series $S'(s \mid \boldsymbol{t})$ as
${ \lvert S'_n(s\mid \boldsymbol{t}) - \sum_{k=1}^\infty (-z_{(k)})^q \exp\{-z_{(k)}(s+\Lambda(\boldsymbol{t}))\}g(z_{(k)}) \rvert} \leq {\lvert S_{n+1} \rvert < \epsilon}$ for all $n>\tilde{n}$ and all $s \in (-r,\infty)$.
This holds for all $\boldsymbol{t}\geq \boldsymbol{t}^*$.
Hence, the first and second order derivative of $\mathcal{L}_{\boldsymbol{t}}(s)$ at $s=0$ and so $\mu_1(\boldsymbol{t}), \mu_2(\boldsymbol{t})$ as well as the $\text{RFV}(\boldsymbol{t})$ exist for all $\boldsymbol{t}\geq\boldsymbol{t}^*$. 
Note that $\boldsymbol{t}^*$ is not claimed to be the lowest threshold for which both conditional moments do exist.

\paragraph{Notes on the adaption of the proof to discrete piecewise constant time-varying frailty:}
As mentioned at the end of Section \ref{Shapes}, the proof may be adapted to certain time-varying (shared) frailty models.
Consider a piecewise constant time-varying (shared) frailty model.
The frailty $Z(t) = Z_q$, if $t \in [t_{q-1}, t_{q})$.
The frailty RVs $Z_1, \dots, Z_Q$ are discrete.
The ordered support of $Z_q$ is $\{z_{q,(1)}, z_{q,(2), \dots}\}$. 
The power of the support of $Z_q$ is $K_q$ which might be infinity.
The pmf of $Z_Q$ is denoted $g_Q$.
Let $\Lambda_q = \sum_{j=1}^J \int_{t_{q-1}}^{t_{q}} \lambda^{(j)}_0(t)dt$,
$\Lambda_Q(\boldsymbol{t}) = \sum_{j=1}^J \int_{t_{Q-1}}^{t_j} \lambda^{(j)}_0(u)du$ for $\boldsymbol{t}\geq t_{Q-1}$, and $\Lambda_Q \in \{\Lambda_Q(\boldsymbol{t}) \mid \boldsymbol{t}\geq t_{Q-1}\}$.
Further we denote $c(z) = \sum_{k_1 = 1}^{K_1} \dots \sum_{k_{Q-1}=1}^{K_{Q-1}}  \exp\{-\sum_{q=1}^{Q-1} z_{q,(k)}\Lambda_q\} P (Z_1=z_{1,(k_1)},\dots, Z_{Q-1}=z_{Q-1,(k_{Q-1})} \mid Z_Q = z)$.
Note that $c(z) = P(\boldsymbol{T}> {t_{Q-1} \mid  Z_Q = z})$ and hence lies between zero and one.
For $\boldsymbol{t} \geq t_{Q-1}$, $P(\boldsymbol{T} > \boldsymbol{t} \mid Z_Q = z)$ simplifies to $\exp\{- z \Lambda_Q(\boldsymbol{t})\} c(z)$.
Additionally, $P(\boldsymbol{T} \geq \boldsymbol{t}) = \sum_{k_1}^{K_1} \dots \sum_{k_{Q}}^{K_{Q}}  P(\boldsymbol{T} \geq \boldsymbol{t} \mid Z_1 = z_{1,(k_1)}, \dots, Z_Q=z_{Q,(k_Q)}) P(Z_1 = z_{1,(k_1)}, \dots,Z_{Q-1}=z_{Q-1,(k_{Q-1})} \mid  Z_Q=z_{Q,(k_Q)})g_Q(Z_Q=z_{Q,(k_Q)})$.
Which, with the help of the former definitions, simplifies to $ \sum_{k=1}^{K_Q}\exp\{- z_{Q,(k)} \Lambda_Q(t)\} g_Q(z_{Q,(k)}) c(z_{Q,(k)})$, for $\boldsymbol{t}\geq t_{Q-1}$.

Two adaptions have to be made in the proof in order to extend it to the given time-varying (shared) frailty scenario.
Firstly, the proof is thought to step in at $t_{Q-1}$, i.e. each element of the vector $\boldsymbol{t}$ is assumed to be larger than $t_{Q-1}$.
Secondly, the derivatives have to be taken w.r.t. $\Lambda_Q$ instead of $\Lambda$.
The ``$M_k$`s" are merely adapted in notation by being either $z_{Q,k}$ or $z_{Q,k}g_Q(z_{Q,k})$. 
To show this we reconsider $P(Z_Q=z \mid \boldsymbol{T} \geq \boldsymbol{t})$ which we adapt in notation to the time-varying scenario by setting it to $g_Q(z,\Lambda_Q(\boldsymbol{t}))$.

The probability of $Z_Q=z$ given survival up to $\boldsymbol{t}\geq t_Q$ can be expressed as
\begin{align}
g_Q(z,\Lambda_Q(\boldsymbol{t})) &= \frac{P(\boldsymbol{T} \geq \boldsymbol{t}  \mid  Z_Q=z) g_Q(z)}{P(\boldsymbol{T} \geq \boldsymbol{t})} \nonumber \\
&= \frac{\exp\{- z \Lambda_Q(\boldsymbol{t})\} c(z) g_Q(z)}{ \sum_{k=1}^{K_Q}\exp\{- z_{Q,(k)} \Lambda_Q(\boldsymbol{t})\} g_Q(z_{Q,(k)}) c(z_{Q,(k)})} \label{gztpreQ}.
\end{align}
Now, $c(z) g_Q(z)$ can be redefined to be $g(z)$ and \eqref{gztpreQ} is the time-varying counterpart of \eqref{gztpre}.
As $0\leq c(z) \leq 1$, the remainder of the proof is unaffected.
Hence, if $z_{Q,(1)} > 0$ the $\text{RFV}$ approaches zero, whereas the $\text{RFV}$ approaches infinity for $z_{Q,(1)}=0$.



\end{appendices}





\end{document}